\documentclass[10pt,journal,letterpaper,compsoc,onecolumn]{IEEEtran}
\pdfoutput=1
\ifCLASSOPTIONcompsoc
  \usepackage[nocompress]{cite}
\else
\fi

\usepackage{amsfonts}
\usepackage{amsmath}
\usepackage{amsthm}
\usepackage{amssymb}
\usepackage{graphicx}
\usepackage{algorithm}
\usepackage{algorithmic}
\usepackage{setspace}

\newtheorem{theorem}{Theorem}

\newtheoremstyle{example}{\topsep}{\topsep}
{}
{}
{\bfseries}
{}
{  }
{\thmname{#1}\thmnumber{ #2.}\thmnote{ (#3)}}
\theoremstyle{example}
\newtheorem{example}{Example}

\newtheorem{lemma}[theorem]{Lemma}

\begin{document}
\doublespacing
\title{Minimal memory requirements for pearl-necklace encoders of quantum convolutional codes}

\author{Monireh~Houshmand,~Saied~Hosseini-Khayat,~and~Mark~M.~Wilde% <-this % stops a space
\IEEEcompsocitemizethanks{\IEEEcompsocthanksitem M.~Houshmand and S.~Hosseini-Khayat are with
the Department of Electrical Engineering, Ferdowsi University of Mashhad, Iran. \protect\\
E-mail: monirehhoushmand@gmail.com; saied.hosseini@gmail.com
\IEEEcompsocthanksitem M.~M.~Wilde is
a postdoctoral fellow with the School of Computer Science, McGill University,
Montreal, Qu\'{e}bec, Canada. \protect\\
E-mail: mark.wilde@mcgill.ca}% <-this % stops a space
\thanks{}}

\IEEEcompsoctitleabstractindextext{%
\begin{abstract}
%\boldmath
One of the major goals in quantum information processing is to reduce the overhead
associated with the practical implementation of quantum protocols,
and often, routines for quantum error correction
account for most of this overhead. A
particular technique for quantum error correction that may be useful for protecting a
stream of quantum information is quantum
convolutional coding. The encoder for a quantum convolutional code has a representation as
a convolutional encoder or as a \textquotedblleft pearl-necklace\textquotedblright\ encoder. In the pearl-necklace representation, it has
not been particularly clear in the research literature how much quantum memory
such an encoder would require for implementation. Here, we offer an algorithm
that answers this question. The algorithm first constructs a weighted,
directed acyclic graph where each vertex of the graph corresponds to a gate
string in the pearl-necklace encoder, and each path through the graph represents a path through non-commuting gates in the encoder. We show that
the weight of the longest path through the graph is equal to the minimal amount of memory
needed to implement the encoder. A dynamic programming search through this
graph determines the longest path. The running time for the construction of the graph
and search through it is quadratic in the number of gate strings in the pearl-necklace
encoder.
\end{abstract}
% not want either math or citations to appear in the abstract.

% Note that keywords are not normally used for peer review papers.
}

\maketitle

% is not selected <OR> if conference mode is selected - because compsoc
% conference papers position the abstract like regular (non-compsoc)
\IEEEdisplaynotcompsoctitleabstractindextext

% \fi
% \IEEEpeerreviewmaketitle
Quantum information science~\cite{book2000mikeandike} is an interdisciplinary field combining quantum physics, mathematics and computer science.
Quantum computers give dramatic speedups over classical ones for tasks such as integer factorization~\cite{Shor-factorization} and
database search~\cite{grover-search}. Two parties can also securely agree on a secret key by exploiting certain features of quantum mechanics~\cite{qkd}.

A quantum system intracts with its environment, and this interaction
inevitably alters the state of the quantum system, which causes loss
of information encoded in it. Quantum error correction~\cite{qecbook,thesis97gottesman,book2000mikeandike} offers a way to combat this noise---it
is the fundamental theory underpinning the practical
realization of quantum computation and quantum communication. The routines associated with it will account for most
of the overhead in the implementation of several practical quantum protocols. Thus, any reduction in the overhead or
resources for implementing quantum error correction should aid in building a practical quantum system.
One example of such a resource is the size of a quantum memory needed to implement the routines of
quantum error correction.

A quantum convolutional code is a particular quantum code that protects a stream of quantum information communicated
over a quantum channel~\cite{PhysRevLett.91.177902,ieee2007forney}. These codes
are inspired by their classical counterparts~\cite{book1999conv} and inherit many of their properties:
 they admit a mathematical description in terms of a parity
check matrix of binary polynomials or binary rational functions and have a memory structure.
They also have low-complexity encoding and decoding circuits and an efficient maximum likelihood error estimation procedure helps estimate errors under the assumption that they are transmitted over a memoryless channel~\cite{arx2007poulin,PhysRevLett.91.177902,arxiv2004olliv}.
%,TL10}
%They are the correct generalization to the
%quantum domain of their classical analogues, and hence
%inherit their most important properties.

%Quantum error correction is the fundamental theory underpinning the practical
%realization of quantum computation and quantum communication
%\cite{qecbook,thesis97gottesman,book2000mikeandike}. The routines associated with it will
%account for most of the overhead in the implementation of a quantum computer. Thus, any reduction in the
%overhead or resources for implementing quantum error correction should aid in
%building a practical quantum computer. One example of such a resource is the
%size of a quantum memory needed to implement the routines of quantum error correction.
%
%A particular method for quantum error correction is quantum convolutional
%coding~\cite{PhysRevLett.91.177902,ieee2007forney}, an approach useful
%for protecting a stream of quantum information. The
%advantage of a quantum convolutional code when compared to a quantum block
%code is an improvement in the trade-off of performance
%with complexity~\cite{ieee2007forney}. Additionally, fast decoding algorithms exist
%for these quantum codes~\cite{arx2007poulin}%,TL10}
%, ensuring that computation
%delays and error suppression delays should not accumulate too fast---a property
%essential in any fault-tolerant implementation.

One representation of the encoder for a quantum convolutional code has a simple
form~\cite{PhysRevLett.91.177902}. It consists of a single unitary repeatedly
applied to a stream of quantum data---we call such a form for the encoder a
\textit{convolutional encoder} (see Figure~\ref{fig:qcc-pearl-necklace}(a)).
An important practical concern for the implementation of an encoder is the
amount of quantum storage or memory it requires. The representation of the
encoder in the convolutional form allows one to determine this quantity in a
straightforward manner: it is equal to the number of qubits that are fed back
into the next iteration of the unitary that acts on the stream. For example,
the convolutional encoder in
Figure~\ref{fig:example-technique}(c)\ requires three memory qubits for
implementation. Ollivier and Tillich pursued this approach for encoding in their early work
on quantum convolutional codes~\cite{PhysRevLett.91.177902,arxiv2004olliv},
and more recently, Poulin, Ollivier, and Tillich exploited this approach in
their construction of quantum turbo codes~\cite{arx2007poulin}.
They randomly generated and filtered
Clifford unitaries at random to act as the convolutional encoders for the
constituent quantum convolutional codes of a quantum turbo code. In this
case, it was straightforward to determine the memory required to implement a
quantum turbo code because they represented the two constituent encoders for
the quantum turbo code in the convolutional form.%}

%TCIMACRO{\FRAME{ftbpFU}{6.0502in}{1.9432in}{0pt}{\Qcb{Two different
%representations of the encoder for a quantum convolutional code. (a)
%Representation of the encoder as a convolutional encoder. (b) Representation
%of the encoder as a pearl-necklace encoder. The numbering at the inputs of the pearl-necklace encoder indicates our
%convention for frame indexing.}}{\Qlb{fig:qcc-pearl-necklace}}{qccandpearlnecklace.pdf}%
%width 6.0502in;  height 1.9432in;  depth 0pt;  original-width 11.0004in;
%original-height 3.4999in;  cropleft "0";  croptop "1";  cropright "1";
%filename 'qccandpearlnecklace.pdf';file-properties "XNPEU";}}}%
%BeginExpansion
\begin{figure}
[ptb]
\begin{center}
\includegraphics[
natheight=3.499900in,
natwidth=11.000400in,
width=6.0502in
]%
{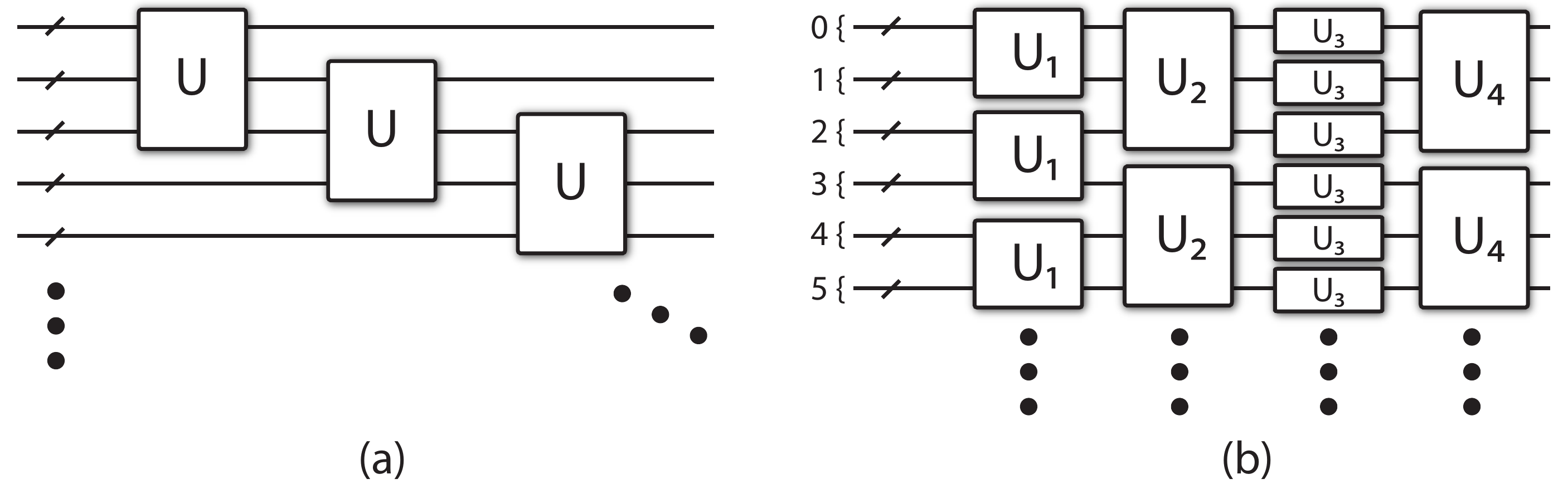}%
\caption{Two different representations of the encoder for a quantum
convolutional code. (a) Representation of the encoder as a convolutional
encoder. (b) Representation of the encoder as a pearl-necklace encoder. The
numbering at the inputs of the pearl-necklace encoder indicates our
convention for frame indexing.}%
\label{fig:qcc-pearl-necklace}%
\end{center}
\end{figure}
%EndExpansion
An alternate representation for the encoder of a quantum convolutional code
consists of several \textquotedblleft strings\textquotedblright\ of the same unitary
applied to the quantum data stream (see Figure~\ref{fig:qcc-pearl-necklace}%
(b)). This representation of the encoder is known as a \textit{pearl-necklace
encoder} due to its striking similarity to a pearl necklace (each string of
unitaries corresponds to one bead of a pearl necklace in this analogy).
Ollivier and Tillich coined this term and realized the importance of this
structure early on~\cite{PhysRevLett.91.177902,arxiv2004olliv}, while Grassl
and R\"{o}tteler (GR) later constructed detailed encoding algorithms for encoding
quantum convolutional codes with pearl-necklace encoders~\cite{isit2006grassl}%
. The algorithm consists of a sequence
of elementary encoding operations. Each of these
elementary encoding operations corresponds to a gate string in the pearl-necklace encoder.
 Grassl and R\"{o}tteler then showed how to produce a quantum convolutional
code from two dual-containing classical binary convolutional codes (much like
the Calderbank-Shor-Steane or CSS approach~\cite{PhysRevA.54.1098,PhysRevLett.77.793}%
) and then constructed a pearl-necklace encoder for the produced code~\cite{ieee2007grassl}.
Later work demonstrated how to
produce an entanglement-assisted quantum convolutional code from two
\textit{arbitrary} classical binary convolutional codes and, in some cases, it
is possible to construct a pearl-necklace encoder for the produced quantum
code\cite{arx2007wildeCED,arx2007wildeEAQCC,arx2008wildeUQCC,arx2008wildeGEAQCC,pra2009wilde}. The advantage of the Grassl-R\"{o}tteler (GR)
and the subsequent entanglement-assisted approach is that the quantum code
designer can choose high-performance classical convolutional codes to import
for use as high-performance quantum convolutional codes.

The representation of a GR pearl-necklace encoder as a convolutional encoder
was originally somewhat unclear, but a recent paper demonstrated how to
translate between these two representations\footnote{The author of
Ref.~\cite{PhysRevA.79.062325}\ called a convolutional encoder a
\textquotedblleft quantum shift register circuit\textquotedblright\ to make
contact with the language of classical shift registers, but the two terms are
essentially interchangeable.}~\cite{PhysRevA.79.062325}. There, the author
exploited notions from linear system theory to show how convolutional encoders
realize the transformations in the GR pearl-necklace
encoders.\footnote{Perhaps Ref.~\cite{PhysRevA.79.062325} is the most explicit
work to show why quantum convolutional codes are in fact \textquotedblleft
convolutional.\textquotedblright} An important contribution of
Ref.~\cite{PhysRevA.79.062325}\ was to clarify the notion of quantum memory in
a GR pearl-necklace encoder, given that such a notion is not explicitly clear
in the pearl-necklace representation. In fact, Ref.~\cite{PhysRevA.79.062325}
demonstrates that a particular convolutional encoder for the
Forney-Grassl-Guha code~\cite{ieee2007forney} requires fives frame of memory
qubits, whereas an earlier analysis of Grassl and R\"{o}tteler suggested that
this encoder would require only two frames~\cite{isit2006grassl}.

The goal of the present paper is to outline an algorithm that computes the memory requirements of a GR pearl-necklace
encoder.\footnote{Ref.~\cite{PhysRevA.79.062325} suggested a formula as an
upper bound on the memory requirements for the Grassl-R\"{o}tteler pearl-necklace encoder of a Calderbank-Shor-Steane\ (CSS) code, but subsequent
analysis demonstrates this upper bound does not hold for all encoders. The
algorithm in the present paper is able to determine the exact memory
requirements of a given Grassl-R\"{o}tteler pearl-necklace encoder.} Our approach considers a class of potential
convolutional encoders that realize the same transformation that a GR
pearl-necklace encoder does. The ideas are in the same
spirit as those in Ref.~\cite{PhysRevA.79.062325}, but the approach here is
different. The algorithm to compute the memory requirements then consists of two parts:

\begin{enumerate}
\item It first constructs a weighted, directed acyclic graph where each vertex
of the graph corresponds to a string of gates in the GR pearl-necklace
encoder. The graph features a directed edge from one vertex to another if the two
corresponding gate strings do not commute, and the weight of a
directed edge depends on the degrees of the two corresponding gate strings.
 %the number of frames upon which a single gate in the first string acts.
 Thus, the graph details paths through non-commuting gates                   in the
pearl-necklace encoder. The complexity for constructing this graph is
quadratic in the number of gate strings in the GR pearl-necklace encoder.

\item We show that the longest path of the graph corresponds
to the minimal amount of memory that a convolutional encoder requires, and the
final part of the algorithm finds this longest path through the graph with
dynamic programming~\cite{cormen}. This final part has complexity linear in
the number of vertices and edges in the graph (or, equivalently, quadratic in the number of gate strings in the pearl-necklace encoder)
because the graph is directed
and acyclic.
\end{enumerate}
%We apply this algorithm to the GR pearl-necklace encoders for CSS\ quantum
%convolutional codes for now, and a later work will address the general case of
%non-CSS\ quantum convolutional codes. The contribution here should be of use
%in the practical realization of encoders for quantum convolutional codes.
In this paper, we focus on encoding CSS\ quantum convolutional codes, for
which each elementary operation corresponds to a string of CNOT\ gates in a
pearl-necklace encoder (see Section~VI of Ref.~\cite{PhysRevA.79.062325}). A later work addresses the general case of
non-CSS\ quantum convolutional codes~\cite{general-paper}.
We begin with a particular pearl-necklace encoder of a quantum convolutional code and determine the minimal amount of
quantum memory needed to implement it as a convolutional encoder.

We structure this work as follows. First we review some basic concepts
from quantum mechanics.
%in particular,
%some notions that Grassl and R\"{o}tteler introduced in
%Refs.~\cite{isit2006grassl,ieee2006grassl,ieee2007grassl}.
Section~\ref{sec:def-not}\ then establishes some definitions and notation that
we employ throughout this paper. Our main contribution begins in
Section~\ref{sec:memory}. We first determine the memory requirements for some
simple examples of pearl-necklace encoders. We then build up to more
complicated examples, by determining the memory required for convolutional
encoders with CNOT gates that are unidirectional, unidirectional in the
opposite direction, and finally with arbitrary direction. The direction is
with respect to the source and target qubits of the CNOT gates in the
convolutional encoder (for example, the convolutional encoder in
Figure~\ref{fig:example-1}(b) is unidirectional). The final section of the paper concludes with a summary and
suggestions for future research.

\section{Quantum states and gates}
The basic data unit in a quantum computer is the qubit. A \emph{qubit} is a unit vector in a two dimensional Hilbert space, $\mathcal{H}_2$ for which a
particular basis, denoted by $\left\vert 0\right\rangle$,
 $\left\vert 1\right\rangle$, has been fixed. The basis states $\left\vert 0\right\rangle$ and $\left\vert 1\right\rangle$ are quantum
analogues of classical 0 and 1 respectively. Unlike classical bits, qubits can be in a superposition of $\left\vert 0\right\rangle$ and $\left\vert 1\right\rangle$  such as $a\left\vert 0\right\rangle+b\left\vert 1\right\rangle$  where $a$ and $b$ are complex numbers such that
$|a|^2 + |b|^2 = 1$. If such a superposition is measured with respect to the basis $\left\vert 0\right\rangle$,$\left\vert 1\right\rangle$, then
$\left\vert 0\right\rangle$ is observed with probability $|a|^2$ and $\left\vert 1\right\rangle$ is observed with probability $|b|^2.$

An $n$-qubit register is a quantum system whose state space is $\mathcal{H}_2^n$. Given
the computational basis $\{\left\vert 0\right\rangle,\left\vert 1\right\rangle\}$ for $\mathcal{H}_2$, the basis states of this register are in the following set:
\[\{\left\vert i_1\right\rangle \otimes \left\vert i_2\right\rangle \otimes \cdots \otimes
 \left\vert i_n\right\rangle; i_1,i_2,\cdots,i_n=0,1                      \},\]
or equivalently
\[\{   \left\vert i_1i_2\cdots i_n\right\rangle ;  i_1,i_2,\cdots,i_n=0,1              \}.\]
The state $\left\vert \psi \right\rangle$ of an $n$-qubit register is a vector in $2^n$-dimensional space:
\[\left\vert \psi \right\rangle=\sum_{i_1,i_2,\cdots,i_n=0,1}a_{i_1,i_2,\cdots,i_n}\left\vert i_1\right\rangle \otimes
 \left\vert i_2\right\rangle \otimes \cdots \otimes
 \left\vert i_n\right\rangle,\]
 where
 \[\sum_{i_1,i_2,\cdots,i_n=0,1}|a_{i_1,i_2,\cdots,i_n}|^2=1.\]

 The phenomenon of \emph{quantum entanglement}~\cite{book2000mikeandike}, which has no classical analogue,
has been recognized as an important physical resource in many areas
of quantum computation and quantum information science.
 A multi-qubit quantum state $\left\vert \psi \right\rangle$ is said to be entangled if it cannot be written as
the tensor product $\left\vert \psi \right\rangle=\left\vert \phi_1 \right\rangle \otimes \left\vert \phi_2\right \rangle $ of two pure states. For example, the EPR pair
shown below is an entangled quantum state:
\[\left\vert \Phi  \right\rangle=(\left\vert 00 \right\rangle+\left\vert 11 \right\rangle)/\surd{2}.\]

In other words, in the case of an entangled state, the qubits are linked in a way such that one cannot describe the quantum state of a constituent of the system independent of its counterparts, even if the individual qubits are spatially separated.

As with classical circuits, quantum operations can be performed by networks of gates.
 Every quantum gate is a linear transformation represented by a unitary matrix, defined on an $n$-qubit Hilbert
space.
A matrix $U$ is \emph{unitary} if
$UU^{\dagger} = I$, where $U^{\dagger}$ is the conjugate transpose of the matrix $U$. Since any unitary operation has an inverse, any quantum gate is reversible,
meaning that given the state of a set of output qubits, it is possible to determine the state of its corresponding
set of input qubits.

Some examples
of useful single-qubit gates are the elements of the Pauli set $\Pi.$ The set $\Pi=\{I,X,Y,Z\}$ consists of the Pauli operators:
\[
I\equiv%
\begin{bmatrix}
1 & 0\\
0 & 1
\end{bmatrix}
,\ X\equiv%
\begin{bmatrix}
0 & 1\\
1 & 0
\end{bmatrix}
,\ Y\equiv%
\begin{bmatrix}
0 & -i\\
i & 0
\end{bmatrix}
,\ Z\equiv%
\begin{bmatrix}
1 & 0\\
0 & -1
\end{bmatrix}
.
\]
$I$ is the identity transformation, $X$ is a bit flip (NOT), $Z$ is a phase flip operation, and $Y$
is a combination of both. Two other important single-qubit transformations are the Hadamard gate $H$ and phase gate $P$ where
\[
H\equiv%
\frac{1}{\sqrt{2}}\begin{bmatrix}
1 & 1\\
1 & -1
\end{bmatrix},\]
\[
P\equiv
\begin{bmatrix}
1 & 0\\
0 & i
\end{bmatrix}
.
\]
The $n$-qubit Pauli group $\Pi^{n}$ is defined as $n$-fold tensor products
of Pauli operators:%
\begin{equation}
\Pi^{n}=\left\{  e^{i\phi}A_{1}\otimes\cdots\otimes A_{n}:\forall j\in\left\{
1,\ldots,n\right\}  ,\ \ A_{j}\in\Pi,\ \ \phi\in\left\{  0,\pi/2,\pi
,3\pi/2\right\}  \right\}  .
\end{equation}
The controlled-NOT gate (CNOT) gate is a two-qubit gate. The first qubit serves as a control and the second as a target. CNOT performs the NOT operation
 on the target qubit if the
control qubit is $\left\vert 1\right\rangle$ and otherwise leaves it unchanged. In other words, the second output is the XOR of the target and control qubit.
 The matrix representation of the CNOT gate is
\[
\text{CNOT}\equiv%
\begin{bmatrix}
1 & 0& 0&0\\
0 & 1&0&0\\
0&0&0&1\\
0&0&1&0
\end{bmatrix}.\]
Figure~\ref{fig:description-of-cnot}(a) shows the circuit representation of the CNOT gate. The non-commutativity of CNOT gates is the most important concept needed to understand this paper. Two CNOT gates do not commute if the index of the source qubit of
one is the same as the index of the target of the other. Figure \ref{fig:description-of-cnot} (b) and (c) show this fact. If the input to both circuits is $ \left\vert a\right\rangle \otimes \left\vert b\right\rangle \otimes \left\vert c\right\rangle (a,b,c=0,1),$ the output of the circuit depicted in Figure \ref{fig:description-of-cnot} (b) is $ \left\vert a\right\rangle \otimes \left\vert a\oplus b\right\rangle \otimes \left\vert b \oplus c\right\rangle,$ while the output of the circuit depicted in Figure \ref{fig:description-of-cnot} (c) is $ \left\vert a\right\rangle \otimes \left\vert a\oplus b\right\rangle \otimes \left\vert a \oplus b \oplus c\right\rangle.$ Therefore the third qubits of two circuits are different when $a=1.$
 \begin{figure}
[ptb]
\begin{center}
\includegraphics[
natheight=3.499900in,
natwidth=11.000400in,
width=6.0502in
]%
{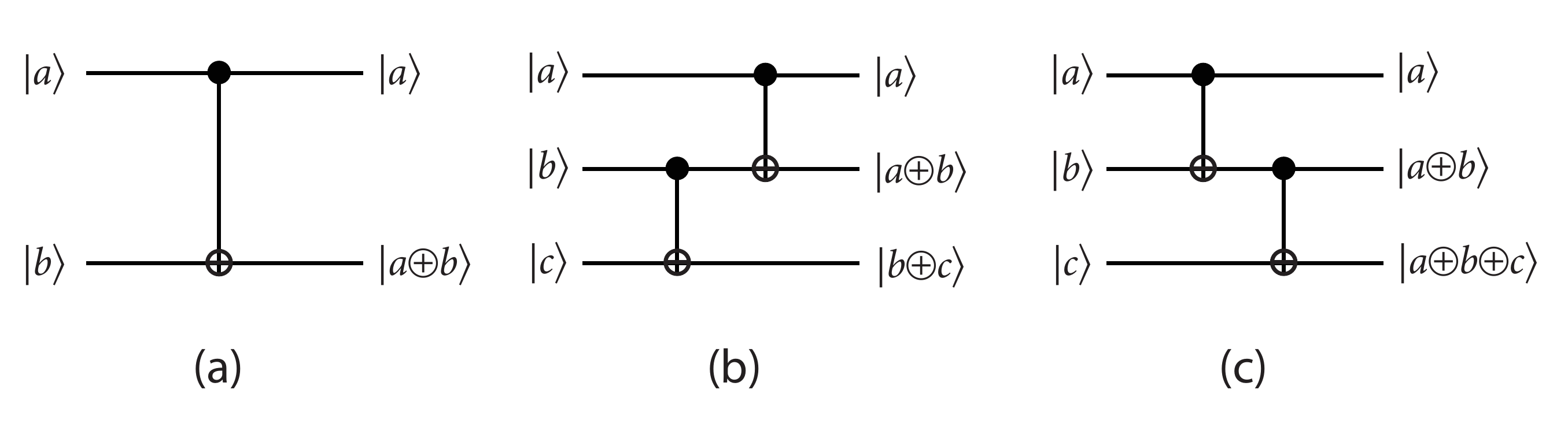}
\caption{(a) the circuit representation of CNOT gate. (b) shows that when index of the source qubit of
one gate is the same as the index of the target of the other, two CNOTs do not commute.}
\label{fig:description-of-cnot}%
\end{center}
\end{figure}

Thus, there are two
kinds of non-commutativity in which we are interested for two gate strings in
a pearl-necklace encoder:

\begin{enumerate}
\item \textit{Source-target non-commutativity} occurs when the index of each
source qubit in the first gate string is the same as the index of each target
qubit in the second gate string. This type of non-commutativity
occurs in the following two gate strings:%
\[
\overline{\text{CNOT}}(i,j)(D^{l_1})\ \overline{\text{CNOT}}(k,i)(D^{l_2}),
\]
where we order the gate strings from left to right.
%(this is the natural order
%for a matrix of the form $B\left(  D\right)  $ in (\ref{eq:GR-form}) that
%applies elementary column operations).

\item \textit{Target-source non-commutativity} occurs when the index of each
target qubit in the first gate string is the same as the index of each source
qubit in the second gate string. For example, this type of non-commutativity
occurs in the following two gate strings:%
\[
\overline{\text{CNOT}}(i,j)(D^{l_1})\ \overline{\text{CNOT}}(j,k)(D^{l_2}).
\]

\end{enumerate}

More generally, if $U$ is a gate that operates on a single qubit where
\[
U\equiv%
\begin{bmatrix}
u_{00} & u_{01}\\
u_{10}&u_{11}
\end{bmatrix},\]
then the controlled-U gate is a gate that operates on two qubits in such a way that the first qubit serves as a control and the second as a target.
 This gate performs the unitary $U$ on the target qubit if the
control qubit is $\left\vert 1\right\rangle$ and otherwise leaves it unchanged. The matrix representation of the controlled-U gate is:
\[
\text{Controlled-U}\equiv%
\begin{bmatrix}
1 & 0& 0 &0\\
0&1& 0& 0\\
0& 0& u_{00}& u_{01}\\
0& 0& u_{10}& u_{11}
\end{bmatrix}.\]

A special type of unitary matrix which is often used in the encoding circuit of a quantum error correction code is called a
Clifford operation~\cite{book2000mikeandike}. A Clifford operation $U$\ is one that preserves elements
of the Pauli group under conjugation:\ $A\in\Pi^{n}\Rightarrow UAU^{\dag}%
\in\Pi^{n}$. The CNOT\ gate, the Hadamard gate $H$, and the phase gate $P$
suffice to implement any unitary matrix in the Clifford group
\cite{thesis97gottesman}.

\section{Definitions and notation}

\label{sec:def-not}We first establish some definitions and notation before
proceeding with the main development of the paper.

Our convention for numbering the frames in a pearl-necklace encoder is from
\textquotedblleft top\textquotedblright\ to \textquotedblleft
bottom.\textquotedblright\ In contrast, our convention for numbering the
frames upon which the unitary of a convolutional encoder acts is from
\textquotedblleft bottom\textquotedblright\ to \textquotedblleft
top.\textquotedblright\ Figure~\ref{fig:qcc-pearl-necklace}(b)\ illustrates
the former convention for a pearl-necklace encoder, while Figure~\ref{fig:example-1}%
(b) illustrates the later convention for a convolutional encoder. These conventions
are useful for our analysis later on.

We now establish conventions for indexing gates in pearl-necklace encoders.
Let $s_{k}$ and $t_{k}$ denote the frame index of the respective source and
target qubits of a gate in the $k^{\text{th}}$ gate string of a pearl-necklace
encoder. For example, consider the pearl-necklace encoder in
Figure~\ref{fig:example-technique}(a) that has two gate strings. The index
$k=1$ for the left gate string, and $k=2$ for the right gate string. The
second CNOT\ gate in the $k=1$ gate string has $s_{1}=1$ and $t_{1}=1$. The
third CNOT\ gate in the $k=2$ gate string has $s_{2}=2$ and $t_{2}=3$.

We also require some conventions for indexing gates in a convolutional
encoder. Let $\sigma_{k}$ and $\tau_{k}$ denote the frame index of the
respective source and target qubits of the $k^{\text{th}}$ gate in a
convolutional encoder. For example, consider the convolutional encoder in
Figure~\ref{fig:example-1}(b). The third gate from the left has $k=3$,
$\sigma_{3}=2$, and $\tau_{3}=0$.

Whether referring to a pearl-necklace encoder or a convolutional encoder, the
notation CNOT$(i,j)\left(  s,t\right)  $ denotes a CNOT gate from qubit $i$ in
frame $s$ to qubit $j$ in frame $t$. We employ this notation extensively in
what follows. The notation $\overline{\text{CNOT}}\left(  i,j\right)  \left(
D^{l}\right)  $ refers to a string of gates in a pearl-necklace encoder and
denotes an infinite, repeated sequence of CNOT gates from qubit $i$ to qubit
$j$ in every frame where qubit $j$ is in a frame delayed by $l$. For example,
the left string of gates in Figure~\ref{fig:example-technique}(a) corresponds
to $\overline{\text{CNOT}}\left(  1,2\right)  \left(  1\right)  $, while the
right string of gates corresponds to $\overline{\text{CNOT}}\left(
1,3\right)  \left(  D\right)  $.

The following two Boolean functions are useful later on in our algorithms for
computing memory requirements:%
\begin{align*}
&  \text{Source-Target}\left(  \overline{\text{CNOT}}(a_{1},b_{1})(D^{l_{1}%
}),\overline{\text{CNOT}}(a_{2},b_{2})(D^{l_{2}})\right)  ,\\
&  \text{Target-Source}\left(  \overline{\text{CNOT}}(a_{1},b_{1})(D^{l_{1}%
}),\overline{\text{CNOT}}(a_{2},b_{2})(D^{l_{2}})\right)  .
\end{align*}
The first function takes two gate strings $\overline{\text{CNOT}}(a_{1}%
,b_{1})(D^{l_{1}})$ and $\overline{\text{CNOT}}(a_{2},b_{2})(D^{l_{2}})$ as
input. It returns TRUE if $\overline{\text{CNOT}}(a_{1},b_{1})(D^{l_{1}})$ and
$\overline{\text{CNOT}}(a_{2},b_{2})(D^{l_{2}})$ have source-target
non-commutativity (i.e., $a_{1}=b_{2}$) and returns FALSE otherwise. The
second function also takes two gate strings $\overline{\text{CNOT}}%
(a_{1},b_{1})(D^{l_{1}})$ and $\overline{\text{CNOT}}(a_{2},b_{2})(D^{l_{2}})$
as input. It returns TRUE if $\overline{\text{CNOT}}(a_{1},b_{1})$ and
$\overline{\text{CNOT}}(a_{2},b_{2})$ have target-source non-commutativity
(i.e., $a_{1}=b_{2}$) and returns FALSE otherwise.

The following succession of $N$\ gate strings realizes a pearl-necklace encoder:
\begin{equation}
\overline{\text{CNOT}}\left(  a_{1},b_{1}\right)  \left(  D^{l_{1}}\right)
\ \overline{\text{CNOT}}\left(  a_{2},b_{2}\right)  \left(  D^{l_{2}}\right)
\ \cdots\ \overline{\text{CNOT}}\left(  a_{N},b_{N}\right)  \left(  D^{l_{N}%
}\right)  .
\end{equation}
Consider the $j^{\text{th}}$ gate string
$\overline{\text{CNOT}}\left(  a_{j},b_{j}\right)  \left(  D^{l_{j}}\right)  $
in the above succession of $N$ gate strings. It is important to consider the
gate strings preceding this one that have source-target non-commutativity with it,
target-source non-commutativity with it, non-negative degree, and negative degree.
This leads to four different subsets $\mathcal{S}_{j}^{+}$, $\mathcal{S}%
_{j}^{-}$, $\mathcal{T}_{j}^{+}$, and $\mathcal{T}_{j}^{-}$ that we define as
follows:
\begin{align*}
\mathcal{S}_{j}^{+}  &  =\{i\mid\text{Source-Target}(\overline{\text{CNOT}%
}\left(  a_{i},b_{i}\right)  \left(  D^{l_{i}}\right)  ,\overline{\text{CNOT}%
}\left(  a_{j},b_{j}\right)  \left(  D^{l_{j}}\right)  )=\text{TRUE}%
,i\in\{1,2,\cdots,j-1\},l_{i}\geq0\},\\
\mathcal{S}_{j}^{-}  &  =\{i\mid\text{Source-Target}(\overline{\text{CNOT}%
}\left(  a_{i},b_{i}\right)  \left(  D^{l_{i}}\right)  ,\overline{\text{CNOT}%
}\left(  a_{j},b_{j}\right)  \left(  D^{l_{j}}\right)  )=\text{TRUE}%
,i\in\{1,2,\cdots,j-1\},l_{i}<0\},\\
\mathcal{T}_{j}^{+}  &  =\{i\mid\text{Target-Source}(\overline{\text{CNOT}%
}\left(  a_{i},b_{i}\right)  \left(  D^{l_{i}}\right)  ,\overline{\text{CNOT}%
}\left(  a_{j},b_{j}\right)  \left(  D^{l_{j}}\right)  )=\text{TRUE}%
,i\in\{1,2,\cdots,j-1\},l_{i}\geq0\},\\
\mathcal{T}_{j}^{-}  &  =\{i\mid\text{Target-Source}(\overline{\text{CNOT}%
}\left(  a_{i},b_{i}\right)  \left(  D^{l_{i}}\right)  ,\overline{\text{CNOT}%
}\left(  a_{j},b_{j}\right)  \left(  D^{l_{j}}\right)  )=\text{TRUE}%
,i\in\{1,2,\cdots,j-1\},l_{i}<0\}.
\end{align*}
The first subset $\mathcal{S}_{j}^{+}$ consists of all the non-negative-degree gate strings preceding gate $j$ that have source-target non-commutativity with it.
The second subset $\mathcal{S}_{j}^{-}$ consists of all the negative-degree gate strings preceding gate $j$ that have source-target non-commutativity with it.
The third subset $\mathcal{T}_{j}^{+}$ consists of all the non-negative-degree gate strings preceding gate $j$ that have target-source non-commutativity with it.
The fourth subset $\mathcal{T}_{j}^{-}$ consists of all the negative-degree gate strings preceding gate $j$ that have target-source non-commutativity with it. We use these subsets extensively in what follows.

\section{Memory requirements for pearl-necklace encoders}

\label{sec:memory}\begin{figure}[ptb]
\begin{center}
\includegraphics[
natheight=8.340300in,
natwidth=13.094100in,
width=5.0518in
]{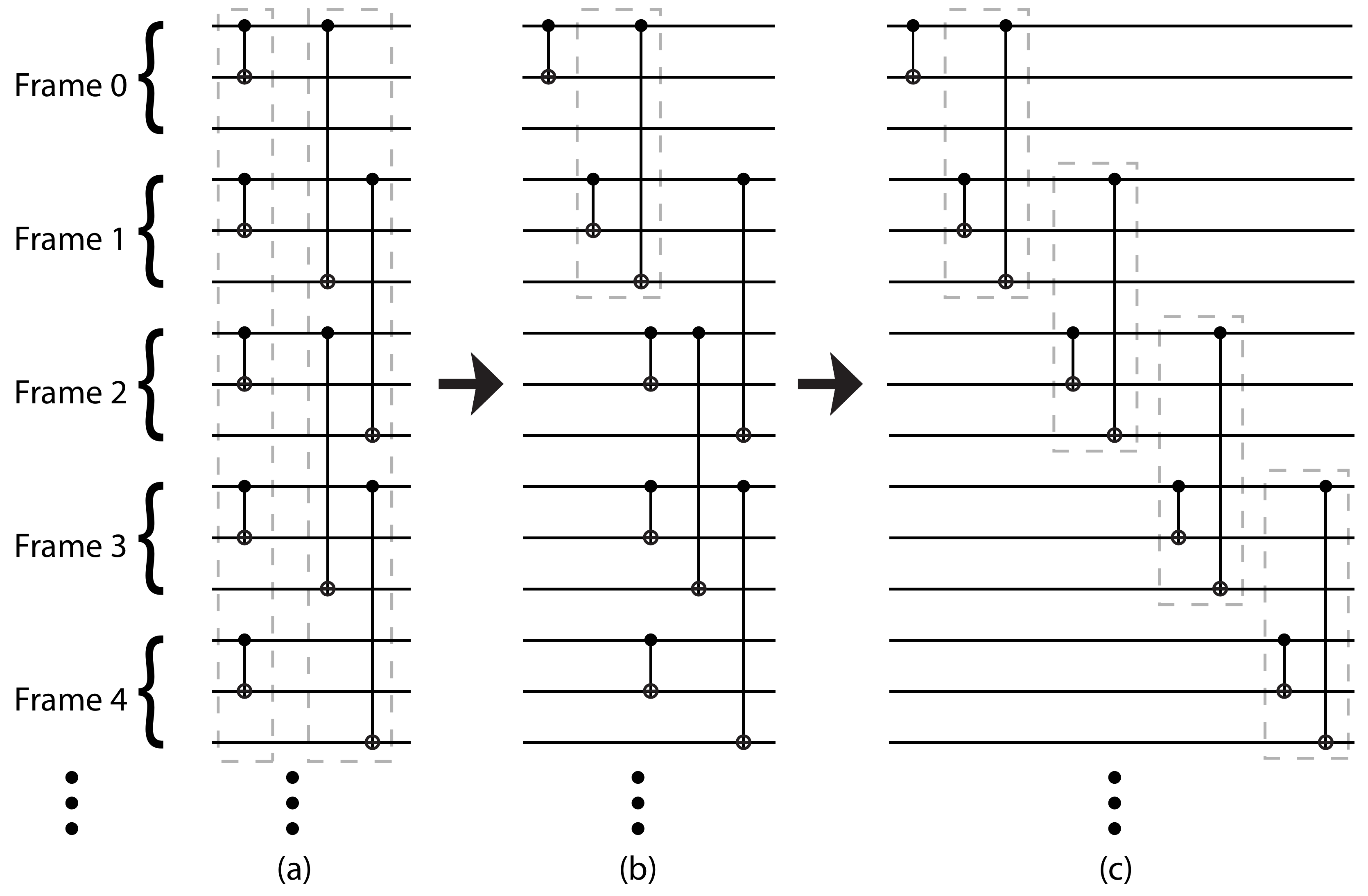}
\end{center}
\caption{Simple depiction of the rearrangement of a pearl-necklace encoder
into a convolutional encoder (note that the technique is more complicated than
depicted for more general cases). (a) The pearl-necklace encoder consists of
the gate strings $\overline{\text{CNOT}}\left(  1,2\right)  \left(  1\right)
\ \overline{\text{CNOT}}\left(  1,3\right)  \left(  D\right)  $. (b) The
rearrangement of the first few gates by shifting the gates below the first three
to the right. (c) A convolutional encoder realization of the pearl-necklace encoder in (a). The repeated application of the procedure in (b) realizes a
convolutional encoder from a pearl-necklace encoder.}
\label{fig:example-technique}
\end{figure}The first step in determining the memory requirements for a GR
pearl-necklace encoder is to rearrange it as a convolutional encoder. There
are many possible correct candidates for the convolutional encoder
(\textquotedblleft correct\textquotedblright\ in the sense that they encode
the same code), but there is a realization that uses a minimal amount of
memory qubits. This idea of rearrangement is in the same spirit as some of the
original ideas of Ollivier and Tillich where there they were trying to
determine non-catastrophic
encoders~\cite{arxiv2004olliv,PhysRevLett.91.177902}, but here we explicitly
apply them to the GR pearl-necklace encoder for the purpose of determining memory
requirements. In order to make a convolutional encoder, we must first find a
set of gates consisting of a single gate for each gate string
in the pearl-necklace encoder such that all of its remaining gates
commute with this set. Then we can shift all the gates remaining in the
pearl-necklace encoder to the right and infinitely repeat this operation on
the remaining gates. Figure~\ref{fig:example-technique} shows a simple example
of the rearrangement of a pearl-necklace encoder $\overline{\text{CNOT}%
}(1,2)\left(  1\right)  \overline{\text{CNOT}}(1,3)(D)$ into a convolutional encoder.

\subsection{The source-target constraint and the target-source constraint}

\label{sec:constraints}We begin by explaining some constraints that apply to
convolutional encoders formed from primitive pearl-necklace encoders.
First consider a pearl-necklace encoder that is a succession of $m$ CNOT gate
strings:
\[
\overline{\text{CNOT}}\left(  a_{1},b_{1}\right)  \left(  D^{l_{1}}\right)
\ \overline{\text{CNOT}}\left(  a_{2},b_{2}\right)  \left(  D^{l_{2}}\right)
\ \cdots\ \overline{\text{CNOT}}\left(  a_{m},b_{m}\right)  \left(  D^{l_{m}%
}\right)  .
\]
Suppose that all the gate strings in the above succession commute with each
other, in the sense that $a_{i}\neq b_{j}$ for all $i\neq j$. Then the
candidates for a convolutional encoder are members of the following set $M$:%
\begin{equation}
M\equiv\left\{  \text{CNOT}(a_{1},b_{1})(s_{1},t_{1})\ \cdots\ \text{CNOT}%
(a_{m},b_{m})(s_{m},t_{m}):t_{i}=s_{i}+l_{i},\ i\in\{1,\dots,m\},\ s_{i}%
\in\left\{  0\right\}  \cup\mathbb{N}\right\}  , \label{eq:M}%
\end{equation}
where $\mathbb{N}=\left\{  1,2,\ldots\right\}  $. All members of $M$ are
correct choices for a convolutional encoder because they produce the required
encoding and because all the remaining gates in the pearl-necklace encoder commute with a particular element of $M$ in all cases. Thus,
there is no constraint on each frame index $s_{i}$ of the source qubit of the
$i^{\text{th}}$ CNOT\ gate.

Now suppose that two CNOT gates in the pearl-necklace encoder do not commute
with each other. Recall that this non-commutativity occurs in two ways:

1)\ Source-target non-commutativity occurs in the following two gate strings:%
\begin{equation}
\overline{\text{CNOT}}\left(  i,j\right)  \left(  D^{l_{1}}\right)
\ \overline{\text{CNOT}}\left(  k,i\right)  \left(  D^{l_{2}}\right)  ,
\label{eq:source-target-encodings}%
\end{equation}
where $j\neq k$.
Potential candidates for a convolutional encoder belong to the following set
$M$:%
\[
M\equiv\left\{  \text{CNOT}(i,j)(s_{1},t_{1})\ \text{CNOT}(k,i)(s_{2}%
,t_{2}):t_{1}=s_{1}+l_{1},\ t_{2}=s_{2}+l_{2},\ s_{1},s_{2}\in\left\{
0\right\}  \cup\mathbb{N}\right\}  ,
\]
though some choices in the set $M$ may not be correct because they ignore the
non-commutativity of the gate strings in (\ref{eq:source-target-encodings}).
In order for the convolutional encoder to be correct, we
should choose the frame indices $s_{1}$ and $t_{2}$ such that all the gates in
the gate string $\overline{\text{CNOT}}\left(  i,j\right)  \left(  D^{l_{1}%
}\right)  $ that remain after CNOT$(i,j)(s_{1},t_{1})$ commute with the gate
CNOT$(k,i)(s_{2},t_{2})$. Otherwise, the chosen convolutional encoder implements
the transformation in (\ref{eq:source-target-encodings}) in the opposite order. An an example, Figure~\ref{fig:example-explanation-source-target} shows the gate strings $\overline{\text{CNOT}}\left(  1,2\right)  \left(  D^{0}\right)
 \overline{\text{CNOT}}\left(  3,1\right)  \left(  D^{1}\right).$
 A correct sample candidate $M$ for the encoder is \[M\equiv \text{CNOT}(1,2)(1,1) \text{CNOT}(3,1)(0,1),\]
 which is shown in the figure. It is obvious that the gates remaining after $\text{CNOT}\left(  1,2\right)\left(  1,1\right)$ (the highlighted gates) commute with $\text{CNOT}(3,1)(0,1).$
 % \label{sec:memory}
 \begin{figure}[ptb]
\begin{center}
\includegraphics[
width=1in
]{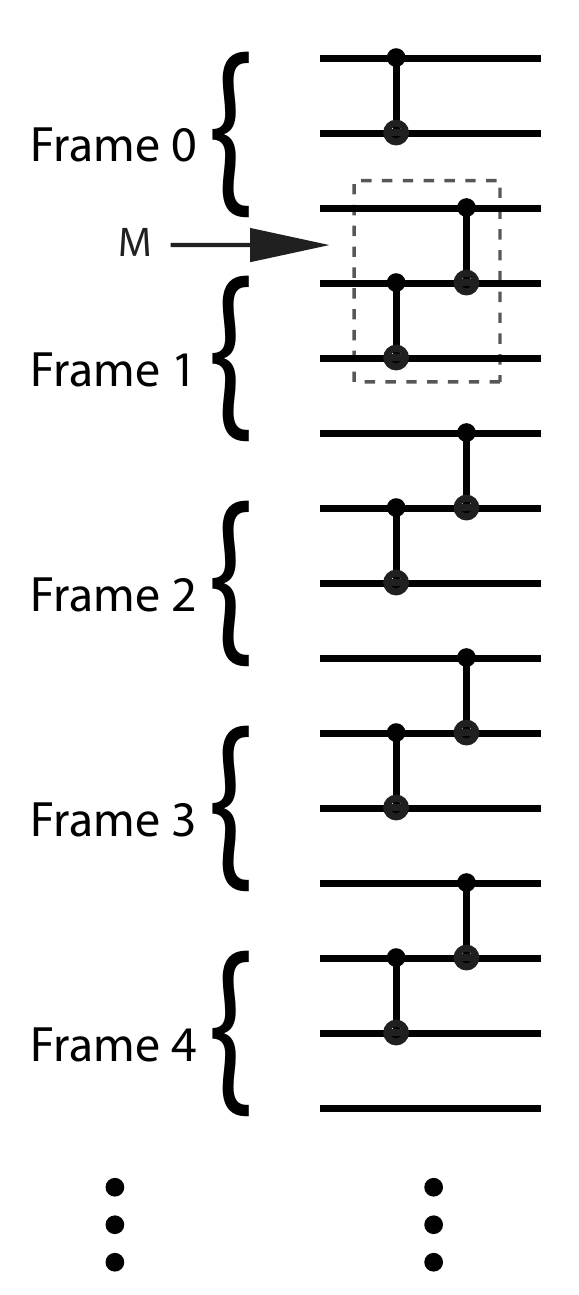}
\end{center}
\caption{A correct sample $M$ for the encoder $\overline{\text{CNOT}}\left(  1,2\right)  \left(  D^{0}\right)
 \overline{\text{CNOT}}\left(  3,1\right)  \left(  D^{1}\right).$ }
\label{fig:example-explanation-source-target}
\end{figure}

The gate $\text{CNOT}(i,j)(t_{2},t_{2}+l_{1})$ is
the only gate in the gate string $\overline{\text{CNOT}}\left(  i,j\right)
\left(  D^{l_{1}}\right)  $ that does not commute with $\text{CNOT}%
(k,i)(s_{2},t_{2})$. Thus, this gate cannot belong to the remaining set of
gates. The set of all gates in the gate string $\overline{\text{CNOT}}\left(
i,j\right)  \left(  D^{l_{1}}\right)  $ remaining after a particular gate CNOT$(i,j)(s_{1}%
,t_{1})$\ is as follows:
\begin{equation}
\{\text{CNOT}(i,j)(s_{1}+d,t_{1}+d):d\in\mathbb{N}\}.\label{eq:ST-set}
\end{equation}
The following inequality determines a restriction
on the source qubit frame index $s_{1}$ such that the gates in the above set both commute with
CNOT$(k,i)(s_{2},t_{2})$ and lead to a correct convolutional encoder:
\begin{equation}
\forall{d}\in\mathbb{N}\ \ \ \ s_{1}+d>t_{2}, \label{eq:ST-set-inequality}
\end{equation}
because these are the remaining gates that we can shift to the right.
Furthermore, the following inequality applies to any correct choice of the first gate in a convolutional encoder
because this gate can be either $\text{CNOT}(i,j)(t_{2},t_{2}+l_{1})$ or any other in the
set in (\ref{eq:ST-set}) that obeys the inequality in (\ref{eq:ST-set-inequality}):
\begin{equation}
s_{1}\geq t_{2}. \label{eq:source-target-constraint}%
\end{equation}
The inequality in (\ref{eq:source-target-constraint}) is the
\textit{source-target constraint} and applies to any correct choice of a convolutional encoder
that implements the transformation in (\ref{eq:source-target-encodings}).

2) The second case is similar to the above case with a few notable changes.
Target-source non-commutativity occurs in the following two gate strings:%
\begin{equation}
\overline{\text{CNOT}}\left(  i,j\right)  \left(  D^{l_{1}}\right)
\ \overline{\text{CNOT}}\left(  j,k\right)  \left(  D^{l_{2}}\right)  .
\label{eq:target-source-encodings}%
\end{equation}
Potential candidates for a convolutional encoder belong to the following set
$M$ where%
\[
M\equiv\left\{  \text{CNOT}(i,j)(s_{1},t_{1})\ \text{CNOT}(j,k)(s_{2}%
,t_{2}):t_{1}=s_{1}+l_{1},\ t_{2}=s_{2}+l_{2},\ s_{1},s_{2}\in\left\{
0\right\}  \cup\mathbb{N}\right\}  ,
\]
though some choices in the set $M$ may not be correct because they ignore the
non-commutativity of the gate strings in (\ref{eq:target-source-encodings}).
In order for the convolutional encoder to be correct, we should choose the frame indices
$t_{1}$ and $s_{2}$ such that the gates in the gate string $\overline
{\text{CNOT}}\left(  i,j\right)  \left(  D^{l_{1}}\right)  $ that remain after
CNOT$(i,j)(s_{1},t_{1})$ commute with CNOT$(j,k)(s_{2},t_{2})$.
Otherwise, the chosen convolutional encoder implements
the transformation in (\ref{eq:target-source-encodings}) in the opposite order. The gate
$\text{CNOT}(i,j)(s_{2}-l_{1},s_{2})$ is the only gate in $\overline
{\text{CNOT}}\left(  i,j\right)  \left(  D^{l_{1}}\right)  $ that does not
commute with $\text{CNOT}(j,k)(s_{2},t_{2})$. Thus, this gate cannot belong to
the remaining set of gates. The set of all gates in the gate string
$\overline{\text{CNOT}}\left(  i,j\right)  \left(  D^{l_{1}}\right)  $
remaining after a particular gate CNOT$(i,j)(s_{1},t_{1})$\ is as follows:
\begin{equation}
\{\text{CNOT}(i,j)(s_{1}+d,t_{1}+d):d\in\mathbb{N}\}.\label{eq:TS-set}
\end{equation}
The following inequality determines a restriction
on the target qubit frame index $t_{1}$ such that the gates in the above set both commute with
CNOT$(j,k)(s_{2},t_{2})$ and lead to a correct convolutional encoder:
\begin{equation}
\forall{d}\in\mathbb{N}\ \ \ \ t_{1}+d>s_{2},\label{eq:TS-set-inequality}
\end{equation}
because these are the remaining gates that we can shift to the right.
Furthermore, the following inequality applies to any correct choice of the first gate in a convolutional encoder
because this gate can be either $\text{CNOT}(i,j)(s_{2}-l_{1},s_{2})$ or any other in the
set in (\ref{eq:TS-set}) that obeys the inequality in (\ref{eq:TS-set-inequality}):
\begin{equation}
t_{1}\geq s_{2}. \label{eq:target-source-constraint}%
\end{equation}
The inequality in (\ref{eq:target-source-constraint}) is the
\textit{target-source constraint} and applies to any correct choice of a convolutional encoder
that implements the transformation in (\ref{eq:target-source-encodings}).

\subsection{Memory requirements for a unidirectional pearl-necklace encoder}

\label{sec:unidirectional-first}We are now in a position to introduce our
algorithms for finding a minimal-memory convolutional encoder that realizes
the same transformation as a pearl-necklace encoder. In this subsection, we
consider the memory requirements for a CSS pearl-necklace encoder with
unidirectional CNOT\ gates (see Figure~\ref{fig:example-1}(b) for an example).
Section~\ref{sec:negative-degree} determines them for a CSS pearl-necklace encoder with
unidirectional CNOT\ gates in the opposite direction, and Section~\ref{sec:arbitrary-CNOT}
determines them for a general CSS pearl-necklace encoder with
CNOT\ gates in an arbitrary direction.

First consider a pearl-necklace encoder that is a sequence of several CNOT
gate strings:%
\[
\overline{\text{CNOT}}\left(  a_{1},b_{1}\right)  \left(  D^{l_{1}}\right)
\ \overline{\text{CNOT}}\left(  a_{2},b_{2}\right)  \left(  D^{l_{2}}\right)
\ \cdots\ \overline{\text{CNOT}}\left(  a_{m},b_{m}\right)  \left(  D^{l_{m}%
}\right)  ,
\]
where all $l_{i}\geq0$ and all the gate strings in the above succession
commute with each other. All members of $M$ in (\ref{eq:M}) are correct
choices for the convolutional encoder, as explained in the beginning of
Section~\ref{sec:constraints}. Though, choosing the same value for each target
qubit frame index $t_{i}$ results in the minimal required memory $L$ where%
\[
L=\max\{l_{1},l_{2},...,l_{m}\}.
\]
A correct, minimal-memory choice for a convolutional encoder is as follows:%
\[
\text{CNOT}(a_{1},b_{1})(l_{1},0)\ \text{CNOT}(a_{2},b_{2})(l_{2}%
,0)\ \cdots\ \text{CNOT}(a_{m},b_{m})(l_{m},0),
\]
where we recall the convention that frames in the convolutional encoder number
from \textquotedblleft bottom\textquotedblright\ to \textquotedblleft
top.\textquotedblright

Now consider two gate strings in a pearl-necklace encoder that have
source-target non-commutativity:%
\begin{equation}
\overline{\text{CNOT}}\left(  i,j\right)  \left(  D^{l_{1}}\right)
\ \overline{\text{CNOT}}\left(  k,i\right)  \left(  D^{l_{2}}\right)  ,
\label{eq:source-target-encodings2}%
\end{equation}
where $l_{1},l_{2}\geq0$. Thus, the source-target constraint in
(\ref{eq:source-target-constraint}) holds for any correct choice of a
convolutional encoder. Choosing $s_{1}=t_{2}$ leads to a minimal-memory
convolutional encoder because any other choice either does not implement the
correct transformation (it violates the source-target constraint) or it uses
more memory than this choice. So a correct, minimal-memory choice for a
convolutional encoder is as follows:%
\[
\text{CNOT}\left(  i,j\right)  \left(  l_{1},0\right)  \ \text{CNOT}\left(
k,i\right)  \left(  l_{1}+l_{2},l_{1}\right)  .
\]
Such a convolutional encoder requires $L$\ frames of memory qubits where%
\[
L=l_{1}+l_{2}.
\]
\begin{figure}[ptb]
\begin{center}
\includegraphics[
natheight=6.106400in,
natwidth=8.753600in,
width=3.5405in
]{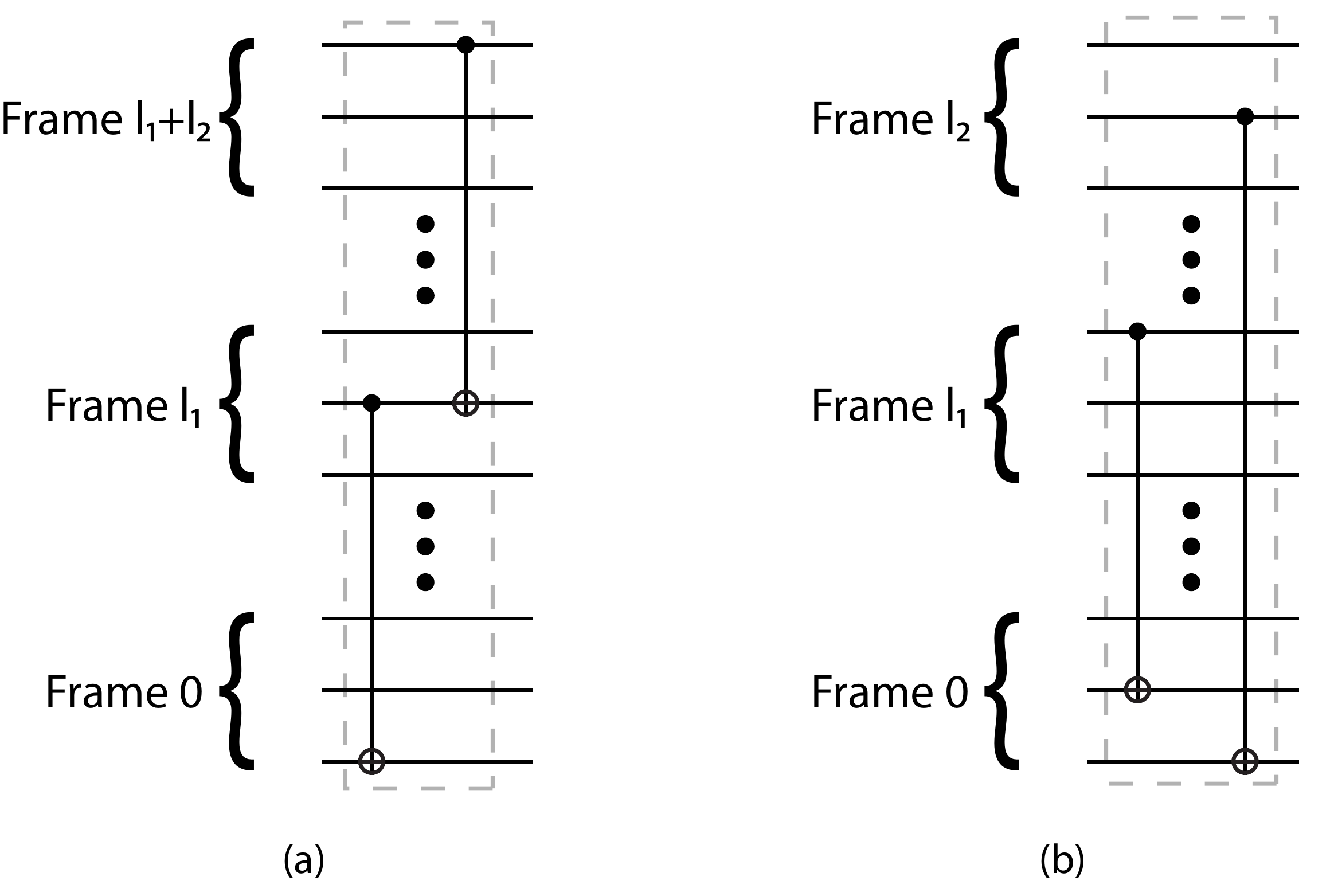}
\end{center}
\caption{Depiction of (a) a minimal-memory convolutional encoder for the gate
strings $\overline{\text{CNOT}}\left(  2,3\right)  \left(  D^{l_{1}}\right)
\ \overline{\text{CNOT}}\left(  1,2\right)  \left(  D^{l_{2}}\right)  $, and
(b) a minimal-memory convolutional encoder for the gate strings $\overline
{\text{CNOT}}\left(  1,2\right)  \left(  D^{l_{1}}\right)  \ \overline
{\text{CNOT}}\left(  2,3\right)  \left(  D^{l_{2}}\right)  $ where $l_1$ and $l_2$ are non-negative.}%
\label{fig:pos-non-comm}%
\end{figure}Figure~\ref{fig:pos-non-comm}(a) depicts a minimal-memory
convolutional encoder for the following gate strings%
\[
\overline{\text{CNOT}}\left(  2,3\right)  \left(  D^{l_{1}}\right)
\ \overline{\text{CNOT}}\left(  1,2\right)  \left(  D^{l_{2}}\right)  ,
\]
where $l_{1}$ and $l_{2}$ are both non-negative.

Consider two gate strings in a pearl-necklace encoder that have target-source
non-commutativity:%
\begin{equation}
\overline{\text{CNOT}}\left(  i,j\right)  \left(  D^{l_{1}}\right)
\ \overline{\text{CNOT}}\left(  j,k\right)  \left(  D^{l_{2}}\right)  ,
\label{eq:source-target-encodings3}%
\end{equation}
where $l_{1},l_{2}\geq0$. Thus, the target-source constraint in
(\ref{eq:target-source-constraint}) holds for any correct choice of a
convolutional encoder. Choosing $t_{1}=t_{2}$ leads to a minimal-memory
convolutional encoder because any other choice either does not implement the
correct transformation (it violates the target-source constraint) or it uses
more memory than this choice. A correct, minimal-memory choice for the
convolutional encoder is as follows:%
\[
\text{CNOT}(i,j)(l_{1},0)\ \text{CNOT}(j,k)(l_{2},0),
\]
and the number $L$ of frames of memory qubits it requires is as follows:%
\[
L=\max\{l_{1},l_{2}\}.
\]
Figure~\ref{fig:pos-non-comm}(b) depicts a minimal-memory convolutional
encoder for the following gate strings%
\[
\overline{\text{CNOT}}\left(  1,2\right)  \left(  D^{l_{1}}\right)
\ \overline{\text{CNOT}}\left(  2,3\right)  \left(  D^{l_{2}}\right)  ,
\]
where both $l_{1}$ and $l_{2}$ are non-negative.

Suppose that two gate strings have both kinds of non-commutativity:%
\[
\overline{\text{CNOT}}(i,j)(D^{l_{1}})\ \overline{\text{CNOT}}(j,i)(D^{l_{2}%
}),
\]
where $l_{1},l_{2}\geq0$. Thus, both constraints in
(\ref{eq:source-target-constraint}) and (\ref{eq:target-source-constraint})
hold for any correct choice of a convolutional encoder. The target-source
constraint in (\ref{eq:target-source-constraint}) holds if the source-target
constraint in (\ref{eq:source-target-constraint}) holds. So it is sufficient
to consider only the source-target constraint in such a scenario.

The above examples prepare us for constructing a minimal-memory convolutional
encoder that implements the same transformation as a pearl-necklace encoder
with unidirectional CNOT gates. Suppose that a pearl-necklace encoder features
the following succession of $N$ gate strings:%
\begin{equation}
\overline{\text{CNOT}}\left(  a_{1},b_{1}\right)  \left(  D^{l_{1}}\right)
\ \overline{\text{CNOT}}\left(  a_{2},b_{2}\right)  \left(  D^{l_{2}}\right)
\ \cdots\ \overline{\text{CNOT}}\left(  a_{N},b_{N}\right)  \left(  D^{l_{N}%
}\right)  ,\label{eq:p-encoding sequence}%
\end{equation}
where all $l_{i}\geq0$. The first gate in the convolutional encoder is
CNOT$\left(  a_{1},b_{1}\right)  \left(  \sigma_{1}=l_{1},\tau_{1}=0\right)
$. For the target indices of each gate $j$ where $2\leq j\leq N$, we should
find the minimal value of $\tau_{j}\ $that satisfies all the source-target and
target-source constraints that the gates preceding it impose. The
inequality in (\ref{eq:ta-cons-1}) applies to the target index of the $j^{\text{th}}$ gate in the
convolutional encoder by applying the source-target
constraint in (\ref{eq:source-target-constraint}):
\begin{align}
\sigma_{i} &  \leq\tau_{j} \,\,\,\,\,\,\forall i\in\mathcal{S}_{j}^{+}, \nonumber\\
\therefore\,\,\tau_{i}+l_{i} &  \leq\tau_{j}\,\,\,\,\,\,\forall i\in\mathcal{S}_{j}^{+},\nonumber\\
\therefore\,\,\max\{\tau_{i}+l_{i}\}_{i\in\mathcal{S}_{j}^{+}} &  \leq\tau
_{j},\label{eq:ta-cons-1}
\end{align}
Recall that the direction of frame numbering in the convolutional encoder is
opposite to the direction of numbering in the pearl-necklace encoder---so the
direction of inequalities are reversed with respect to
(\ref{eq:source-target-constraint}) and (\ref{eq:target-source-constraint}).
The inequality in (\ref{eq:ta-cons-1}) exploits all of the source-target
constraints corresponding to the gates preceding gate $j$ in order to place a
limit on the location of the $j^{\text{th}}$ gate in the convolutional
encoder. The inequality in (\ref{eq:ta-cons-2}) similarly exploits all
of the target-source constraints corresponding to the gates preceding gate
$j$:
\begin{align}
\tau_{i} &  \leq\sigma_{j} \,\,\,\,\,\,\forall i\in\mathcal{T}_{j}^{+}, \nonumber\\
\therefore\,\,\tau_{i}-l_{j} &  \leq\tau_{j}\,\,\,\,\,\,\forall i\in\mathcal{T}_{j}^{+},\nonumber\\
\therefore\,\,\max\{\tau_{i}-l_{j}\}_{i\in\mathcal{T}_{j}^{+}} &  \leq\tau
_{j}.\label{eq:ta-cons-2}%
\end{align}
The following constraint applies to the frame index $\tau_{j}$ of the
target qubit of the $j^{\text{th}}$ gate in the convolutional encoder, by
applying (\ref{eq:ta-cons-1}) and (\ref{eq:ta-cons-2}):%
\[
\tau_{j}\geq\max\{\{\tau_{i}+l_{i}\}_{i\in\mathcal{S}_{j}^{+}},\{\tau_{i}-l_{j}%
\}_{i\in\mathcal{T}_{j}^{+}}\}.
\]
Thus, the minimal value for $\tau_{j}$ that satisfies all the constraints is%
\begin{equation}
\tau_{j}=\max\{\{\tau_{i}+l_{i}\}_{i\in\mathcal{S}_{j}^{+}},\{\tau_{i}%
-l_{j}\}_{i\in\mathcal{T}_{i}^{+}}\}.\label{eq:positive-target-constraint1}%
\end{equation}
Of course, there is no constraint for the frame index $\tau_{j}$ if the gate
string $\overline{\text{CNOT}}(a_{j},b_{j})(D^{l_{j}})$ commutes with all
previous gate strings. Thus, in this case, we choose the frame index $\tau
_{j}$ as follows:%
\begin{equation}
\tau_{j}=0.\label{eq:positive-target-constraint2}%
\end{equation}
A good choice for the frame index $\tau_{j}$ is as follows:%
\begin{equation}
\tau_{j}=\max\{0,\{\tau_{i}+l_{i}\}_{i\in\mathcal{S}_{j}^{+}},\{\tau_{i}%
-l_{j}\}_{i\in\mathcal{T}_{j}^{-}}\},\label{eq:positive-target-constraint}%
\end{equation}
by considering (\ref{eq:positive-target-constraint1}) and
(\ref{eq:positive-target-constraint2}).

\subsubsection{Construction of the commutativity graph}

We introduce the notion of a \emph{commutativity} graph in order to find the
values in (\ref{eq:positive-target-constraint}) for the target qubit
frame indices. The graph is a weighted, directed acyclic graph constructed
from the non-commutativity relations of the gate strings
in~(\ref{eq:p-encoding sequence}). Let $G^{+}$ denote the commutativity graph for a succession of gate strings that have purely non-negative degrees (and
thus where the CNOT\ gates are unidirectional).
Algorithm~\ref{alg:nc-graph-positive} below presents pseudo code for
constructing the commutativity graph~$G^{+}$.

\begin{algorithm}
\caption {Algorithm for determining the commutativity graph $G^{+}$ for
purely non-negative case}
\label{alg:nc-graph-positive}
\begin{algorithmic}
\STATE$N \gets$ Number of gate strings in the pearl-necklace encoder
\STATE Draw a {\bf START} vertex
\FOR{$j := 1$ to $N$}
\STATE Draw a vertex labeled $j$ for the $j^{\text{th}}$ gate string $\overline{\text{CNOT}}(a_{j},b_{j})(D^{l_j})$
\STATE DrawEdge({\bf START}, $j, 0$)
\FOR{$i$ := $1$ to $j-1$}
\IF{$\text{Source-Target}\left( \overline{\text{CNOT}}(a_{i},b_{i})(D^{l_i}),\overline{ \text{CNOT}}(a_{j},b_{j})(D^{l_j})\right) = {\bf TRUE}$}
\STATE DrawEdge($i,j,l_i$)
\ELSIF{Target-Source$\left(\overline{\text{CNOT}}(a_{i},b_{i})(D^{l_i}),\overline{ \text{CNOT}}(a_{j},b_{j})(D^{l_j})\right) = {\bf TRUE}$}
\STATE DrawEdge($i,j,-l_j$)
\ENDIF
\ENDFOR
\ENDFOR
\STATE Draw an {\bf END} vertex
\FOR{$j$ := $1$ to $N$}
\STATE DrawEdge($j$, {\bf END}, $l_j$)
\ENDFOR
\end{algorithmic}
\end{algorithm}

The commutativity graph $G^{+}$ consists of $N$ vertices, labeled
$1,2,\cdots,N$, where the $j^{\text{th}}$ vertex corresponds to the
$j^{\text{th}}$ gate string $\overline{\text{CNOT}}(a_{j},b_{j})(D^{l_{j}})$.
It also has two dummy vertices, named \textquotedblleft
START\textquotedblright\ and \textquotedblleft END.\textquotedblright%
\ DrawEdge$\left(  i,j,w\right)  $ is a function that draws a directed edge
from vertex $i$ to vertex $j$ with an edge weight equal to $w$. A zero-weight
edge connects the START vertex to every vertex, and an $l_{j}$-weight edge
connects every vertex $j$ to the END vertex. Also, an $l_{i}$-weight edge
connects the $i^{\text{th}}$ vertex to the $j^{\text{th}}$\ vertex if%
\[
\text{Source-Target}\left(  \overline{\text{CNOT}}(a_{i},b_{i})(D^{l_{i}%
}),\overline{\text{CNOT}}(a_{j},b_{j})(D^{l_{j}})\right)  =\text{TRUE},
\]
and a $-l_{j}$-weight edge connects the $i^{\text{th}}$ vertex to the
$j^{\text{th}}$\ vertex if%
\[
\text{Target-Source}\left(  \overline{\text{CNOT}}(a_{i},b_{i})(D^{l_{i}%
}),\overline{\text{CNOT}}(a_{j},b_{j})(D^{l_{j}})\right)  =\text{TRUE}.
\]
The commutativity graph $G^{+}$ is an acyclic graph because a directed edge
connects each vertex only to vertices for which its corresponding gate comes later in the pearl-necklace encoder.

The construction of $G^{+}$ requires time quadratic in the number of gate
strings in the pearl-necklace encoder. In
Algorithm~\ref{alg:nc-graph-positive}, the \textbf{if }instruction in the
inner \textbf{for} loop requires constant time $O(1)$. The sum of iterations
of the \textbf{if} instruction in the $j^{\text{th}}$ iteration of the outer
\textbf{for} loop is equal to $j-1$. Thus the running time $T(N)$\ of
Algorithm~\ref{alg:nc-graph-positive} is
\[
T(N)=\sum_{i=1}^{N}{\sum_{k=1}^{j-1}O(1)}=O(N^{2}).
\]

\subsubsection{The longest path gives the minimal memory requirements}

Theorem~\ref{thm:longest-path-is-memory-+} below states that the weight of the
longest path from the START vertex to the END vertex is equal to the minimal memory
required for a convolutional encoder implementation.

\begin{theorem}
\label{thm:longest-path-is-memory-+}The weight $w$\ of the longest path from
the START vertex to END vertex in the commutativity graph $G^{+}$ is equal
to minimal memory $L$ that the convolutional encoder requires.
\end{theorem}

\begin{proof}
We first prove by induction that the weight $w_{j}$ of the longest path from
the START vertex to vertex $j$ in the commutativity graph $G^{+}$ is%
\begin{equation}
w_{j}=\mathbb{\tau}_{j}. \label{eq:Gwlp}%
\end{equation}
A zero-weight edge connects the START vertex to the first vertex, so that
$w_{1}=\tau_{1}=0$. Thus the base step holds for the target index of the first
CNOT gate in a minimal-memory convolutional encoder. Now suppose the property
holds for the target indices of the first $k$ CNOT gates in the convolutional
encoder:%
\begin{equation}
w_{j}=\mathbb{\tau}_{j}\,\,\,\,\,\,\forall j : 1\leq j\leq k. \label{eq:inductive-hypothesis}
\end{equation}
Suppose we add a new gate string $\overline{\text{CNOT}}(a_{k+1},b_{k+1}%
)(D^{l_{k+1}})$ to the pearl-necklace encoder, and
Algorithm~\ref{alg:nc-graph-positive} then adds a new vertex $k+1$ to the
graph $G^{+}$ and the following edges to $G^{+}$:

\begin{enumerate}
\item A zero-weight edge from the START vertex to vertex $k+1$.

\item An $l_{k+1}$-weight edge from vertex $k+1$ to the END vertex.

\item An $l_{i}$-weight edge from each vertex $\{i\}_{i\in\mathcal{S}_{k+1}^{+}}$
to vertex $k+1$.

\item A $-l_{k+1}$-weight edge from each vertex $\{i\}_{i\in\mathcal{T}_{k+1}^{+}%
}$ to vertex $k+1$.
\end{enumerate}
So it is clear that the following relations hold because $w_{k+1}$ is the weight
of the longest path to vertex $k+1$ and from applying (\ref{eq:inductive-hypothesis}):%
\begin{align}
w_{k+1}  &  =\max\{0,\{w_{i}+l_{i}\}_{i\in\mathcal{S}_{k+1}^{+}},\{w_{i}%
-l_{k+1}\}_{i\in\mathcal{T}_{k+1}^{+}}\},\nonumber\\
&  =\max\{0,\{\mathbb{\tau}_{i}+l_{i}\}_{i\in\mathcal{S}_{k+1}^{+}},\{\mathbb{\tau
}_{i}-l_{k+1}\}_{i\in\mathcal{T}_{k+1}^{+}}\}. \label{eq:Gwlp2}%
\end{align}
The inductive proof then follows by applying
(\ref{eq:positive-target-constraint}) and (\ref{eq:Gwlp2}):%
\[
w_{k+1}=\tau_{k+1}.
\]
The proof of the theorem follows by considering the following equalities:%
\begin{align*}
w  &  =\max_{i\in\{1,2,\cdots,N\}}\{w_{i}+l_{i}\}\\
&  =\max_{i\in\{1,2,\cdots,N\}}\{\mathbb{\tau}_{i}+l_{i}\}\\
&  =\max_{i\in\{1,2,\cdots,N\}}\{{{\sigma}_{i}}\}.
\end{align*}
The first equality holds because the longest path in
the graph is the maximum of the weight of the path
to the $i^{\text{th}}$ vertex summed with the weight
of the edge from the $i^{\text{th}}$ vertex to the END vertex. The second equality follows by applying
(\ref{eq:Gwlp}). The final equality follows because ${{\sigma}_{i}}=\tau
_{i}+l_{i}$. The quantity $\max\{{{\sigma}_{i}}\}_{i\in\{1,2,\cdots
,N\}}$ is equal to minimal required memory for a minimal-memory convolutional
encoder because the largest location of a source qubit determines
the number of frames upon which a convolutional encoder with unidirectional CNOT gates acts.
(Recall that we number the frames starting from zero). Thus, the theorem holds.
\end{proof}

The final task is to determine the longest path in $G^{+}$. Finding the
longest path in a general graph is an NP-complete problem, but dynamic
programming finds it on a
weighted, directed acyclic graph in time linear in the number of vertices and edges,
or equivalently, quadratic in the number of gate strings in the pearl-necklace encoder~\cite{cormen}.

\subsubsection{Example of a pearl-necklace encoder with unidirectional CNOT gates}
 We conclude this development with an example.
\begin{example}
\label{ex:positive}Consider the following gate strings in a pearl-necklace
encoder:%
\[
\overline{\text{CNOT}}(2,3)(D)\ \overline{\text{CNOT}}(1,2)(D)\ \overline
{\text{CNOT}}(2,3)(D^{2})\ \overline{\text{CNOT}}(1,2)(1)\ \overline
{\text{CNOT}}(2,1)(D).
\]
All gate strings in the above pearl-necklace encoder have non-negative degree and are thus
unidirectional.
Figure \ref{fig:example-1}(a) draws $G^{+}$ for this pearl-necklace encoder,
after running Algorithm~\ref{alg:nc-graph-positive}. The graph displays all of the
source-target and target-source non-commutativities between gate strings in
the pearl-necklace encoder. The longest path through the graph is
\[
\text{START}\rightarrow3\rightarrow4\rightarrow5\rightarrow\text{END},
\]
with weight equal to three. So the minimal memory for the convolutional encoder
is equal to three frames of
memory qubits. Also from inspecting the graph $G^{+}$,
we can determine the locations for all the target qubit frame indices: $\tau_{1}=0$, $\tau_{2}=1$, $\tau_{3}%
=0$, $\tau_{4}=2$, and $\tau_{5}=2.$ Figure~\ref{fig:example-1}(b) depicts a
minimal-memory convolutional encoder that implements the same transformation as the pearl-necklace encoder.
\end{example}

{%
%TCIMACRO{\FRAME{ftbpFU}{5.047in}{3.3745in}{0pt}{\Qcb{(a) The non-commutative
%graph $G^{+}$ and (b) a minimal-memory convolutional encoder for
%Example~\ref{ex:positive}.}}{\Qlb{fig:example-1}}{example1.pdf}%
%width 5.047in;  height 3.3745in;  depth 0pt;  original-width 11.7269in;
%original-height 6.8338in;  cropleft "0";  croptop "1";  cropright "1";
%cropbottom "0";  filename 'example1.pdf';file-properties "XNPEU";}}}%
%BeginExpansion
\begin{figure}
[ptb]
\begin{center}
\includegraphics[
natheight=6.833800in,
natwidth=11.726900in,
width=5.047in
]%
{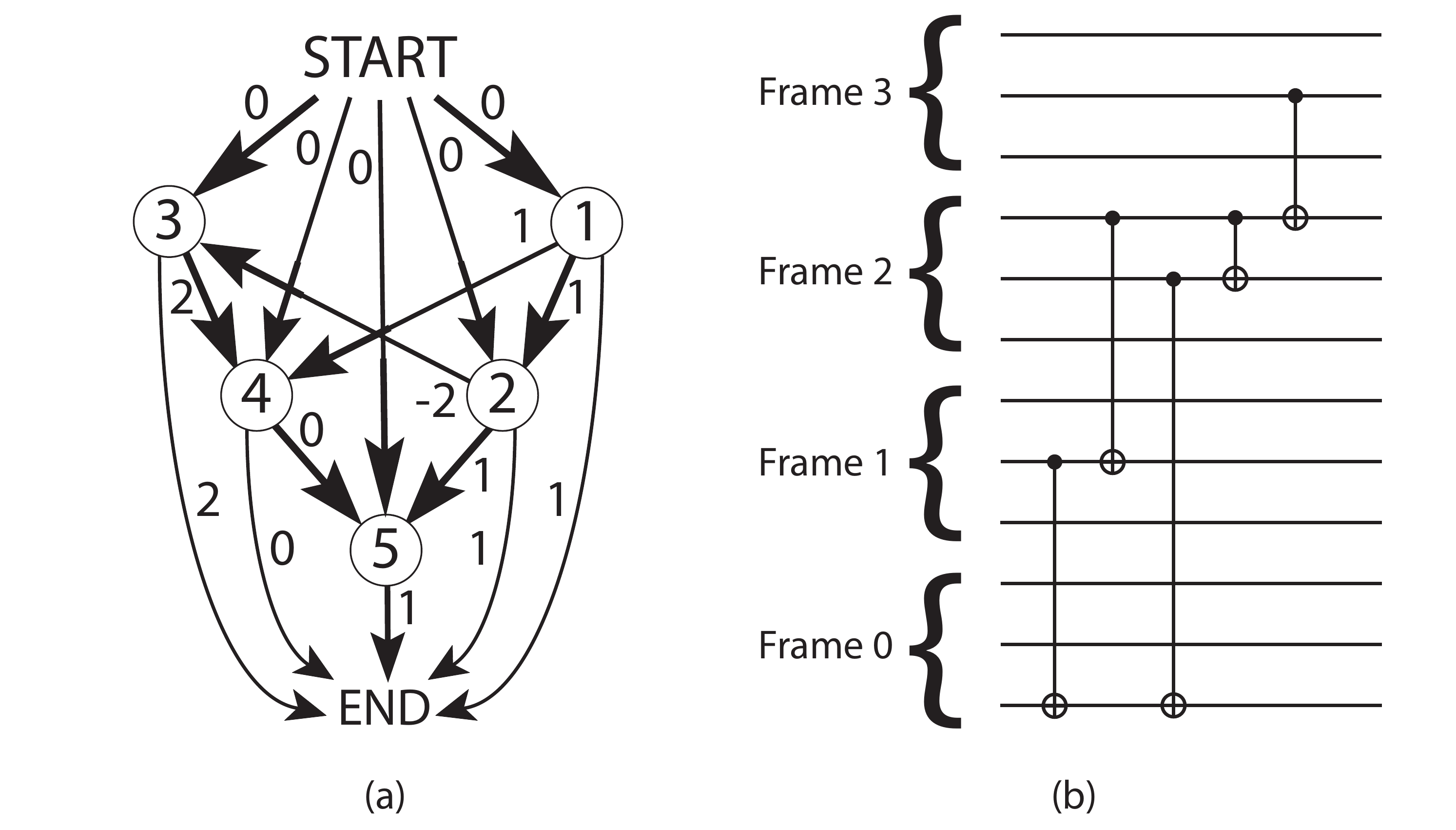}%
\caption{(a) The commutativity graph $G^{+}$ and (b) a minimal-memory
convolutional encoder for Example~\ref{ex:positive}.}%
\label{fig:example-1}%
\end{center}
\end{figure}
%EndExpansion
}

\subsection{Memory requirements for a unidirectional pearl-necklace encoder in
the opposite direction}

\label{sec:negative-degree}In this section, we find a minimal-memory
convolutional encoder that implements the same transformation as
a pearl-necklace encoder with purely non-positive degree
CNOT\ gates. The ideas in this section are similar to those in the previous one.

First consider a pearl-necklace encoder that is a succession of several CNOT
gate strings:%
\[
\overline{\text{CNOT}}\left(  a_{1},b_{1}\right)  \left(  D^{l_{1}}\right)
\ \overline{\text{CNOT}}\left(  a_{2},b_{2}\right)  \left(  D^{l_{2}}\right)
\ \cdots\ \overline{\text{CNOT}}\left(  a_{m},b_{m}\right)  \left(  D^{l_{m}%
}\right)  ,
\]
where all $l_{i}\leq0$ and all the gate strings commute with each other. All
members of $M$ in (\ref{eq:M}) are correct choices for the convolutional
encoder, as explained in the beginning of Section~\ref{sec:constraints}.
But this time, choosing the same value for each source qubit frame index $s_{i}$ results in
the minimal required memory $L$ where%
\[
L=\max\{\left\vert l_{1}\right\vert ,\left\vert l_{2}\right\vert
,...,\left\vert l_{m}\right\vert \}.
\]
A correct choice for a minimal-memory convolutional encoder is%
\[
\text{CNOT}(a_{1},b_{1})(0,\left\vert l_{1}\right\vert )\ \text{CNOT}%
(a_{2},b_{2})(0,\left\vert l_{2}\right\vert )\ \cdots\ \text{CNOT}(a_{m}%
,b_{m})(0,\left\vert l_{m}\right\vert ).
\]

Now consider two gate strings that have source-target non-commutativity:%
\begin{equation}
\overline{\text{CNOT}}\left(  i,j\right)  \left(  D^{l_{1}}\right)
\ \overline{\text{CNOT}}\left(  k,i\right)  \left(  D^{l_{2}}\right)  ,
\end{equation}
where $l_{1},l_{2}\leq0$. Thus, the source-target constraint in
(\ref{eq:source-target-constraint}) holds for any correct choice of a
convolutional encoder. Choosing $s_{1}=s_{2}$ leads to the minimal memory
required for the convolutional encoder because any other choice either does
not implement the correct transformation (it violates the source-target
constraint) or it uses more memory than this choice. A correct choice for a
minimal-memory convolutional encoder is%
\[
\text{CNOT}\left(  i,j\right)  \left(  0,\left\vert l_{1}\right\vert \right)
\ \text{CNOT}\left(  k,i\right)  \left(  0,\left\vert l_{2}\right\vert
\right)  .
\]
Such a convolutional encoder requires $L$\ frames of memory qubits where%
\[
L=\max\{\left\vert l_{1}\right\vert ,\left\vert l_{2}\right\vert \}.
\]
Figure \ref{fig:neg-non-comm}(a) illustrates a minimal-memory convolutional
encoder for the gate strings%
\[
\overline{\text{CNOT}}\left(  3,2\right)  \left(  D^{l_{1}}\right)
\ \overline{\text{CNOT}}\left(  1,3\right)  \left(  D^{l_{2}}\right)  ,
\]
where $l_{1},l_{2}\leq0$.

Now consider two gate strings that have target-source non-commutativity:%
\begin{equation}
\overline{\text{CNOT}}\left(  i,j\right)  \left(  D^{l_{1}}\right)
\overline{\text{CNOT}}\left(  j,k\right)  \left(  D^{l_{2}}\right)  ,
\end{equation}
that where $l_{1},l_{2}\leq0$. The target-source constraint in
(\ref{eq:target-source-constraint}) holds for any correct choice of a
convolutional encoder. Choosing $t_{1}=s_{2}$ leads to a minimal-memory
convolutional encoder because any other choice either does not implement the
correct transformation (it violates the target-source constraint) or it
requires more memory than this choice. A correct choice for a minimal-memory
convolutional encoder is{%
%TCIMACRO{\FRAME{ftbpFU}{3.5405in}{2.4734in}{0pt}{\Qcb{(a) A minimal-memory
%convolutional encoder for the gate strings $\overline{\text{CNOT}}\left(
%3,2\right)  \left(  D^{l_{1}}\right)  \ \overline{\text{CNOT}}\left(
%1,3\right)  \left(  D^{l_{2}}\right)  $, and (b) a minimal-memory
%convolutional encoder for the gate strings $\overline{\text{CNOT}}\left(
%3,2\right)  \left(  D^{l_{1}}\right)  \ \overline{\text{CNOT}}\left(
%2,1\right)  \left(  D^{l_{2}}\right)  $.}}{\Qlb{fig:neg-non-comm}%
%}{negativenoncommutativities.pdf}{\special{ language "Scientific Word";
%type "GRAPHIC";  maintain-aspect-ratio TRUE;  display "USEDEF";
%valid_file "F";  width 3.5405in;  height 2.4734in;  depth 0pt;
%original-width 8.406in;  original-height 5.8531in;  cropleft "0";
%filename 'negativenoncommutativities.pdf';file-properties "XNPEU";}%
%}}%
%BeginExpansion
\begin{figure}
[ptb]
\begin{center}
\includegraphics[
natheight=5.853100in,
natwidth=8.406000in,
width=3.5405in
]%
{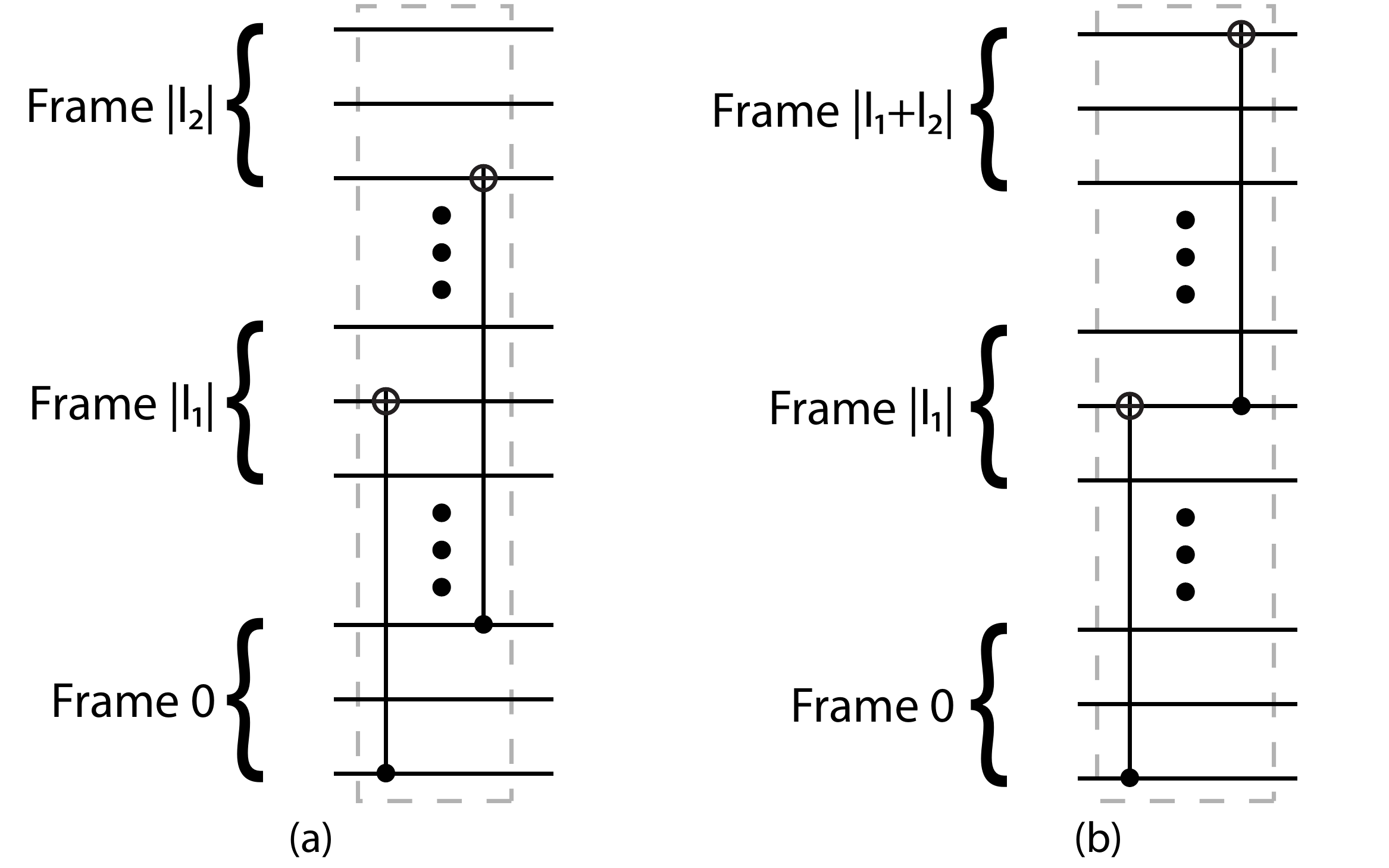}%
\caption{(a) A minimal-memory convolutional encoder for the gate strings
$\overline{\text{CNOT}}\left(  3,2\right)  \left(  D^{l_{1}}\right)
\ \overline{\text{CNOT}}\left(  1,3\right)  \left(  D^{l_{2}}\right)  $, and
(b) a minimal-memory convolutional encoder for the gate strings $\overline
{\text{CNOT}}\left(  3,2\right)  \left(  D^{l_{1}}\right)  \ \overline
{\text{CNOT}}\left(  2,1\right)  \left(  D^{l_{2}}\right)  $ where $l_1$ and $l_2$ are non-positive.}%
\label{fig:neg-non-comm}%
\end{center}
\end{figure}
%EndExpansion
}%
\[
\text{CNOT}\left(  i,j\right)  \left(  0,\left\vert l_{1}\right\vert \right)
\text{CNOT}\left(  k,i\right)  \left(  \left\vert l_{1}\right\vert ,\left\vert
l_{1}+l_{2}\right\vert \right)  ,
\]
with the number $L$\ of frames of memory qubits as follows:%
\[
L=\left\vert l_{1}+l_{2}\right\vert .
\]
Figure \ref{fig:neg-non-comm}(b) shows a minimal-memory convolutional encoder
for the encoding sequence%
\[
\overline{\text{CNOT}}\left(  3,2\right)  \left(  D^{l_{1}}\right)
\ \overline{\text{CNOT}}\left(  2,1\right)  \left(  D^{l_{2}}\right)  ,
\]
where $l_{1}$ and $l_{2}$ are non-positive.

Suppose we have two gate strings that feature both types of
non-commutativity:
\[
\overline{\text{CNOT}}(i,j)(D^{l_{1}})\ \overline{\text{CNOT}}(j,i)(D^{l_{2}%
}).
\]
Thus, both constraints in (\ref{eq:source-target-constraint}) and
(\ref{eq:target-source-constraint}) hold for any correct choice of a
convolutional encoder. The source-target constraint in
(\ref{eq:source-target-constraint}) holds if the target-source constraint in
(\ref{eq:target-source-constraint}) holds when both degrees are non-positive. So
it is sufficient to consider only the target-source constraint in this scenario.

The above examples prepare us for constructing a minimal-memory convolutional
encoder that implements the same transformation as a pearl-necklace encoder
with unidirectional CNOT\ gates (the gates are in the opposite direction of
those in Section~\ref{sec:unidirectional-first}). Suppose that a pearl-necklace encoder features the following succession of $N$ gate strings:%
\begin{equation}
\overline{\text{CNOT}}\left(  a_{1},b_{1}\right)  \left(  D^{l_{1}}\right)
\ \overline{\text{CNOT}}\left(  a_{2},b_{2}\right)  \left(  D^{l_{2}}\right)
\ \cdots\ \overline{\text{CNOT}}\left(  a_{N},b_{N}\right)  \left(  D^{l_{N}%
}\right)  , \label{eq:neg-enc-seq}%
\end{equation}
where $N$ is the number of gate strings and all $l_{i}\leq0$. The first gate
in the convolutional encoder is CNOT$\left(  \sigma_{1}=0,\tau_{1}%
=l_{1}\right)  $. For the source indices of gate $j$ where $2\leq j\leq N$, we
should find the minimal value for $\sigma_{j}$ that satisfies all the
source-target and target-source constraints that the previous gates impose.
The following inequalities apply to the source qubit frame index $\sigma_{j}$
of the $j^{\text{th}}$ gate in the convolutional encoder:
\begin{align}
\sigma_{i}  &  \leq\tau_{j}\,\,\,\,\,\,\forall i\in\mathcal{S}_{j}^{-},\nonumber\\
\therefore\,\,\,\sigma_{i}  &  \leq\sigma_{j}+\left\vert l_{j}\right\vert
\,\,\,\,\,\,\forall i\in\mathcal{S}_{j}^{-},\nonumber\\
\therefore\,\,\,\max\{\sigma_{i}-\left\vert l_{j}\right\vert \}_{i\in\mathcal{S}%
_{j}^{-}}  &  \leq\sigma_{j},\label{eq:so-cons-1}
\end{align}
The inequality in (\ref{eq:so-cons-1}) exploits all of the source-target
constraints corresponding to the gates preceding gate $j$ in order to place a
limit on the location of the $j^{\text{th}}$ gate in the convolutional
encoder. The inequality below similarly exploits all
of the target-source constraints corresponding to the gates preceding gate
$j$:
\begin{align}
\tau_{i}  &  \leq\sigma_{j} \,\,\,\,\,\,\forall i\in\mathcal{T}_{j}^{-},\nonumber\\
\therefore\,\,\,\sigma_{i}+\left\vert l_{i}\right\vert  &  \leq\sigma_{j}
\,\,\,\,\,\,\forall i\in\mathcal{T}_{j}^{-},\nonumber\\
\therefore\,\,\,\max\{\sigma_{i}+\left\vert l_{i}\right\vert \}_{i\in\mathcal{T}_{j}^{-}}
 &  \leq\sigma_{j}. \label{eq:so-cons-2}
\end{align}
The following constraint applies to the frame index $\sigma_{j}$ of the
source qubit of the $j^{\text{th}}$ gate in the convolutional encoder, by
applying (\ref{eq:so-cons-1}) and (\ref{eq:so-cons-2}):%
\[
\sigma_{j}\geq\max\{\{\sigma_{i}-\left\vert l_{j}\right\vert \}_{i\in\mathcal{S}_{j}^{-}%
},\{\sigma_{i}+\left\vert l_{i}\right\vert \}_{i\in\mathcal{T}_{i}^{-}}%
\}.
\]
Thus, the minimal value for $\sigma_{j}$ that satisfies all the constraints is%
\begin{equation}
\sigma_{j}=\max\{\{\sigma_{i}-\left\vert l_{j}\right\vert \}_{i\in
\mathcal{S}_{j}^{-}},\{\sigma_{i}+\left\vert l_{i}\right\vert \}_{i\in
\mathcal{T}_{j}^{-}}\} \label{eq:mas}%
\end{equation}
There is no constraint for the source index $\sigma_{j}$\ if the gate string
$\overline{\text{CNOT}}(a_{j},b_{j})(D^{l_{j}})$ commutes with all previous
gate strings. Thus, in this case, we can choose $\sigma_{j}$\ as follows:%
\begin{equation}
\sigma_{j}=0. \label{eq:h}%
\end{equation}
So, based on (\ref{eq:mas}) and (\ref{eq:h}), a good choice for $\sigma_{j}$
is as follows:%
\begin{equation}
\sigma_{j}=\max\{0,\{\sigma_{i}-\left\vert l_{j}\right\vert \}_{i\in
\mathcal{S}_{j}^{-}},\{\sigma_{i}+\left\vert l_{i}\right\vert \}_{i\in
\mathcal{T}_{j}^{-}}\}\}. \label{eq:negative-source-constraint}%
\end{equation}

\subsubsection{Construction of the commutativity graph for the non-positive
degree case}

We construct a \emph{commutativity} graph $G^{-}$ in order to find the
values in (\ref{eq:negative-source-constraint}). It is again a weighted,
directed acyclic graph constructed from the non-commutativity relations in the
pearl-necklace encoder in (\ref{eq:neg-enc-seq}).
Algorithm~\ref{alg:nc-graph-negative} presents pseudo code for the
construction of the commutativity graph in the non-positive degree case.

\begin{algorithm}
\caption{Algorithm for determining the commutativity graph $G^{-}$ for purely non-positive case}
\label{alg:nc-graph-negative}
\begin{algorithmic}
\STATE$N \gets$ Number of gate strings in the pearl-necklace encoder.
\STATE Draw a {\bf START} vertex.
\FOR{$j := 1$ to $N$}
\STATE Draw a vertex labeled $j$ for the $j^{\text{th}}$ gate string $\overline{\text{CNOT}}(a_{j},b_{j})(D^{l_j})$
\STATE DrawEdge({\bf START}, $j$, 0)
\FOR{$i$ := $1$ to $j-1$}
\IF{Target-Source$\left(\overline{\text{CNOT}}(a_{i},b_{i})(D^{l_i}),\overline{ \text{CNOT}}(a_{j},b_{j})(D^{l_j})\right) = {\bf TRUE}$}
\STATE DrawEdge($i,j,|l_i|$ )
\ELSIF{Source-Target$\left(\overline{\text{CNOT}}(a_{i},b_{i})(D^{l_i}),\overline{ \text{CNOT}}(a_{j},b_{j})(D^{l_j})\right)={\bf TRUE}$}
\STATE DrawEdge($i,j,-|l_j|$)
\ENDIF
\ENDFOR
\ENDFOR
\STATE Draw an {\bf END} vertex.
\FOR{$j$ := $1$ to $N$}
\STATE DrawEdge($j$,{\bf END},$|l_j|$)
\ENDFOR
\end{algorithmic}
\end{algorithm}

The graph $G^{-}$ consists of $N$ vertices, labeled $1,2,\cdots,N$, where
vertex $j$ corresponds to the $j^{\text{th}}$ gate string $\overline
{\text{CNOT}}(a_{j},b_{j})(D^{l_{j}})$. A zero-weight edge connects the START vertex to all
vertices, and an $\left\vert l_{j}\right\vert $-weight edge connects every
vertex $j$ to the END vertex. Also, an $\left\vert l_{i}\right\vert $-weight
edge connects vertex $i$ to vertex $j$ if%
\[
\text{Target-Source}\left(  \overline{\text{CNOT}}(a_{i},b_{i})(D^{l_{i}%
}),\overline{\text{CNOT}}(a_{j},b_{j})(D^{l_{j}})\right)  =\text{TRUE},
\]
and a $-\left\vert l_{j}\right\vert $-weight edge connects vertex $i$ to
vertex $j$ if%
\[
\text{Source-Target}\left(  \overline{\text{CNOT}}(a_{i},b_{i})(D^{l_{i}%
}),\overline{\text{CNOT}}(a_{j},b_{j})(D^{l_{j}})\right)  =\text{TRUE}.
\]
The graph $G^{-}$ is an acyclic graph and its construction complexity is $O(N^{2})$ (similar to the complexity for constructing $G^{+}$). Dynamic programming can find
the longest path in $G^{-}$ in time linear in the number of vertices and edges, or equivalently, quadratic
in the number of gate strings in the pearl-necklace encoder.

\subsubsection{The longest path gives the minimal memory requirements}

We now prove that the weight of the longest path from the START vertex to END vertex in $G^{-}$ is equal to the memory
in a minimal-memory realization of the pearl-necklace encoder in
(\ref{eq:neg-enc-seq}).

\begin{theorem}
The weight of the longest path from the START vertex to END vertex in the graph $G^{-}$
is equal to the minimal memory requirements of the convolutional encoder.
\end{theorem}

\begin{proof}
By similar reasoning as in Theorem~\ref{thm:longest-path-is-memory-+}, the weight
of the longest path from the START vertex to vertex $j$ in the
commutativity graph $G^{-}$ is equal to
\begin{equation}
w_{j}=\sigma_{j} \label{eq:Gwlp3}
\end{equation}
Similar to the proof of Theorem~\ref{thm:longest-path-is-memory-+},
we can prove that the longest path
from the START vertex to END vertex in $G^{-}$ is equal to $\max
\{\mathbb{\tau}_{i}\}_{1\leq i\leq N}$. Thus, it is equal to the minimal required number
of frames of memory qubits.
\end{proof}

\subsubsection{Example of a pearl-necklace encoder with unidirectional CNOT gates in the opposite direction}

The following example illustrates how to find the minimal required memory for
a purely non-positive degree pearl-necklace encoder.\begin{figure}[ptb]
\begin{center}
\includegraphics[ natheight=6.679800in, natwidth=11.306600in,
width=5.5486in ]{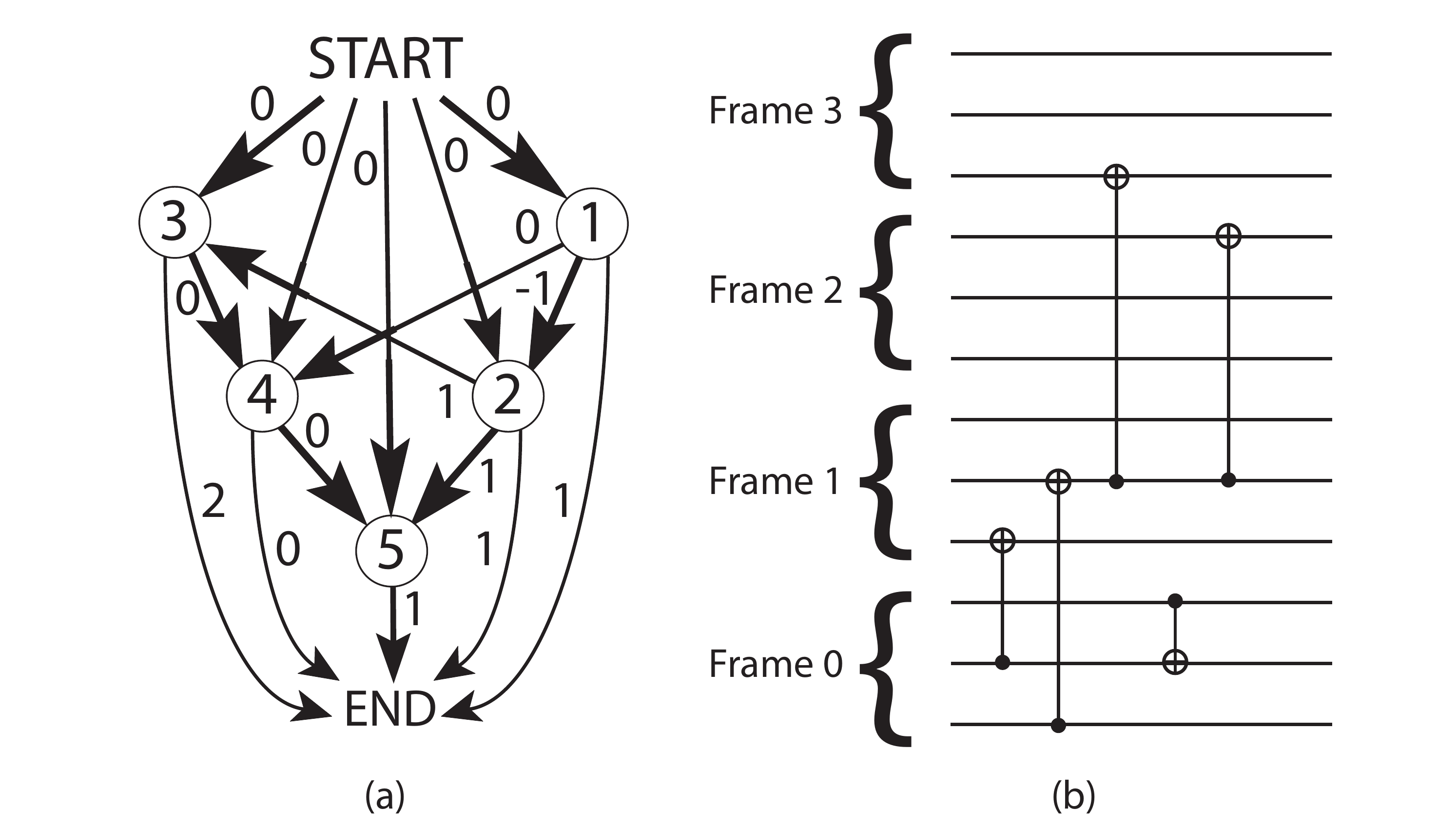}
\end{center}
\caption{(a) The commutativity graph $G^{+}$ and (b) a minimal-memory
convolutional encoder for Example~\ref{ex:negative}.}%
\label{fig:example-2}%
\end{figure}

\begin{example}
\label{ex:negative}Consider the following succession of gate strings in a
pearl-necklace encoder:
\[
\overline{\text{CNOT}}(2,3)(D^{-1})\ \overline{\text{CNOT}}(1,2)(D^{-1}%
)\ \overline{\text{CNOT}}(2,3)(D^{-2})\ \overline{\text{CNOT}}%
(1,2)(1)\ \overline{\text{CNOT}}(2,1)(D^{-1}).
\]
All gates have non-positive powers and thus are unidirectional.
Figure \ref{fig:example-2}(a) illustrates the commutativity graph for this
pearl-necklace encoder. The commutativity graph details all of the source-target
and target-source non-commutativities between gate strings. The longest path in $G^{-}$ is
\[
\text{START}\rightarrow2\rightarrow3\rightarrow\text{END,}%
\]
with its weight equal to three. The memory needed to implement the
convolutional encoder is three frames of memory qubits. From inspecting the
commutativity graph $G^{-}$, we can also determine the locations of the source qubit frame indices: $\sigma_{1}=0$,
$\sigma_{2}=0$, $\sigma_{3}=1$, $\sigma_{4}=0$, and $\sigma_{5}=1$. Figure
\ref{fig:example-2}(b) depicts a minimal-memory convolutional encoder for this example.
\end{example}

\subsection{Memory requirements for an arbitrary CNOT pearl-necklace encoder}
\label{sec:arbitrary-CNOT}
This section is the culmination of the previous two developments in
Sections~\ref{sec:unidirectional-first}\ and~\ref{sec:negative-degree}. Here,
we find a minimal-memory convolutional encoder that implements the same
transformation as a general pearl-necklace encoder with arbitrary CNOT gate strings.

Consider a pearl-necklace encoder that is a succession of several arbitrary
CNOT gate strings:%
\begin{equation}
\overline{\text{CNOT}}\left(  a_{1},b_{1}\right)  \left(  D^{l_{1}}\right)
\ \overline{\text{CNOT}}\left(  a_{2},b_{2}\right)  \left(  D^{l_{2}}\right)
\ \cdots\ \overline{\text{CNOT}}\left(  a_{N},b_{N}\right)  \left(  D^{l_{N}%
}\right)  . \label{eq:mixed}%
\end{equation}

We construct a commutativity graph $G$ in order to determine a
minimal-memory convolutional encoder. This graph is similar to those in
Sections~\ref{sec:unidirectional-first}\ and~\ref{sec:negative-degree}, but it
combines ideas from both developments. In this graph, the weight of the
longest path from the START vertex to vertex $j$ is equal to $\tau_{j}$ when
$l_{j}\geq0$, and it is equal to $\sigma_{j}$ when $l_{j}\leq0$. We consider
the constraints that the gates preceding gate $j$ impose. The constraint
inequalities use the target qubit frame index $\tau_{j}$ when $l_{j}\geq0$ and
use the source qubit frame index $\sigma_{j}$ when $l_{j}<0.$ First consider
the case when $l_{j}\geq0$. The source-target and target-source constraints
that previous gates impose on gate $j$ occur in four different ways, based on
the sign of the involved gate's degree:

\begin{enumerate}
\item There is a source-target constraint for all gates preceding gate $j$ that have non-negative degree and source-target non-commutativity with it:
\begin{align*}
\sigma_{i}  & \leq\tau_{j}\,\,\,\,\,\,\forall i\in\mathcal{S}_{j}^{+}\\
\therefore\,\,\,\tau_{i}+l_{i}  & \leq\tau_{j} \,\,\,\,\,\,\forall i\in\mathcal{S}_{j}^{+}.
\end{align*}

\item There is a source-target constraint for all gates preceding gate $j$ that have negative degree and source-target non-commutativity with it:
\[
\sigma_{i}\leq\tau_{j} \,\,\,\,\,\,\forall i\in\mathcal{S}_{j}^{-}.
\]

\item There is a target-source constraint for all gates preceding gate $j$ that have non-negative degree and target-source non-commutativity with it:
\begin{align*}
\tau_{i}  & \leq\sigma_{j}\,\,\,\,\,\,\forall i\in\mathcal{T}_{j}^{+}\\
\therefore\,\,\,\tau_{i}  & \leq\tau_{j}+l_{j}\,\,\,\,\,\,\forall i\in\mathcal{T}_{j}^{+},\\
\therefore\,\,\,\tau_{i}-l_{j}  & \leq\tau_{j}\,\,\,\,\,\,\forall i\in\mathcal{T}_{j}^{+}.
\end{align*}

\item There is a target-source constraint for all gates preceding gate $j$ that have negative degree and target-source non-commutativity with it:
\begin{align*}
\tau_{i} &  \leq\sigma_{j}\,\,\,\,\,\,\forall i\in\mathcal{T}_{j}^{-}\\
\therefore\,\,\,\sigma_{i}+\left\vert l_{i}\right\vert  &  \leq\tau_{j}+l_{j} \,\,\,\,\,\,\forall i\in\mathcal{T}_{j}^{-},\\
\therefore\,\,\,\sigma_{i}+\left\vert {l_{i}}\right\vert -l_{j} &  \leq\tau_{j}\,\,\,\,\,\,\forall i\in\mathcal{T}_{j}^{-}.
\end{align*}

\end{enumerate}

The graph includes an edge from vertex $i$ to vertex $j$, corresponding to
each of the above constraints. The target qubit frame
index $\tau_{j}$ should satisfy the following inequality, by considering the
above four inequalities:%
\begin{equation}
\max\{\{\tau_{i}+l_{i}\}_{i\in\mathcal{S}_{j}^{+}},\{\sigma_{i}\}_{i\in
\mathcal{S}_{j}^{-}},\{\tau_{i}-l_{j}\}_{i\in\mathcal{T}_{j}^{+}},\{\sigma
_{i}+|l_{i}|-l_{j}\}_{i\in\mathcal{T}_{j}^{-}}\}\leq\tau_{j}.
\end{equation}
Choosing $\tau_{j}$ so that it minimally satisfies the above constraints
results in a minimal usage of memory:%
\begin{equation}
\tau_{j}=\max\{\{\tau_{i}+l_{i}\}_{i\in\mathcal{S}_{j}^{+}},\{\sigma
_{i}\}_{i\in\mathcal{S}_{j}^{-}},\{\tau_{i}-l_{j}\}_{i\in\mathcal{T}_{j}^{+}%
},\{\sigma_{i}+|l_{i}|-l_{j}\}_{i\in\mathcal{T}_{j}^{-}}\}. \label{pos-cons1}%
\end{equation}
There is no constraint for a gate string $\overline{\text{CNOT}}(a_{j}%
,b_{j})(D^{l_{j}})$ that commutes with its previous gates:%
\begin{equation}
\tau_{j}=0. \label{pos-cons-2}%
\end{equation}
Thus choosing $\tau_{j}$ as follows when $l_{j}\geq0$ results in minimal
memory usage, based on (\ref{pos-cons1}) and (\ref{pos-cons-2}):%
\begin{equation}
\tau_{j}=\max\{0,\{\tau_{i}+l_{i}\}_{i\in\mathcal{S}_{j}^{+}},\{\sigma
_{i}\}_{i\in\mathcal{S}_{j}^{-}},\{\tau_{i}-l_{j}\}_{i\in\mathcal{T}_{j}^{+}%
},\{\sigma_{i}+\left\vert l_{i}\right\vert -l_{j}\}_{i\in\mathcal{T}_{j}^{-}%
}\}. \label{eq:tar-mixed-cons}%
\end{equation}

Now we consider the constraints that gates preceding gate $j$ impose on it
when $l_{j}<0$. There are four different non-commutativity constraints based
on the sign of the involved gate's degree:

\begin{enumerate}
\item There is a source-target constraint for all gates preceding gate $j$ that have non-negative degree and source-target non-commutativity with it:
\begin{align*}
\sigma_{i} &  \leq\tau_{j}\,\,\,\,\,\,\forall i\in\mathcal{S}_{j}^{+}\\
\therefore\,\,\,\tau_{i}+l_{i} &  \leq\sigma_{j}+\left\vert {l_{j}}\right\vert \,\,\,\,\,\,\forall i\in\mathcal{S}_{j}^{+},\\
\therefore\,\,\,\tau_{i}+l_{i}-\left\vert {l_{j}}\right\vert  &  \leq\sigma_{j}\,\,\,\,\,\,\forall i\in\mathcal{S}_{j}^{+}.
\end{align*}

\item There is a source-target constraint for all gates preceding gate $j$ that have negative degree and source-target non-commutativity with it:
\begin{align*}
\sigma_{i}  &  \leq\tau_{j}\,\,\,\,\,\,\forall i\in\mathcal{S}_{j}^{-}\\
\therefore\,\,\,\sigma_{i}  &  \leq\sigma_{j}+\left\vert {l_{j}}\right\vert \,\,\,\,\,\,\forall i\in\mathcal{S}_{j}^{-}\\
\therefore\,\,\,\sigma_{i}-\left\vert {l_{j}}\right\vert  &  \leq\sigma_{j}\,\,\,\,\,\,\forall i\in\mathcal{S}_{j}^{-}.
\end{align*}

\item There is a target-source constraint for all gates preceding gate $j$ that have non-negative degree and target-source non-commutativity with it:
\[
\tau_{i}\leq\sigma_{j}\,\,\,\,\,\,\forall i\in\mathcal{T}_{j}^{+}.
\]

\item There is a target-source constraint for all gates preceding gate $j$ that have negative degree and target-source non-commutativity with it:
\begin{align*}
\tau_{i}  &  \leq\sigma_{j}\,\,\,\,\,\,\forall i\in\mathcal{T}_{j}^{+},\\
\therefore\,\,\,\sigma_{i}+\left\vert {l_{i}}\right\vert  &  \leq\sigma_{j}\,\,\,\,\,\,\forall i\in\mathcal{T}_{j}^{+}.
\end{align*}

\end{enumerate}
For similar reasons as above, choosing $\sigma_{j}$ as follows results in
minimal memory usage when $l_{i}<0$:%
\begin{equation}
\sigma_{j}=\max\{0,\{\tau_{i}+l_{i}-\left\vert l_{j}\right\vert\}_{i\in\mathcal{S}_{j}^{+}},\{\sigma
_{i}-\left\vert l_{j}\right\vert\}_{i\in\mathcal{S}_{j}^{-}},\{\tau_{i}\}_{i\in\mathcal{T}_{j}^{+}%
},\{\sigma_{i}+\left\vert l_{i}\right\vert\}_{i\in\mathcal{T}_{j}^{-}%
}\}. \label{eq:sou-mixed-cons}%
\end{equation}

A search through the constructed commutativity graph $G$ can find the values
in (\ref{eq:tar-mixed-cons}) and (\ref{eq:sou-mixed-cons}).
Algorithm~\ref{alg:nc-graph-mixed} below gives the pseudo code for
constructing the commutativity graph $G$. The graph $G$ consists of $N$ vertices,
labeled $1,2,\cdots,N$, and vertex $j$
corresponds to $j^{\text{th}}$ gate string$~\overline{\text{CNOT}}(a_{j}%
,b_{j})(D^{l_{j}})$ in the pearl-necklace encoder.

\begin{algorithm}
\caption{Algorithm for determining the commutativity graph $G$ in mixed case}
\label{alg:nc-graph-mixed}
\begin{algorithmic}
\STATE$N \gets$ Number of gate strings in the pearl-necklace encoder
\STATE Draw a {\bf START} vertex.
\FOR{$j := 1$ to $N$}
\STATE Draw a vertex labeled $j$ for the $j^{\text{th}}$ encoding operation, $\overline{\text{CNOT}}(a_{j},b_{j})(D^{l_j})$
\STATE DrawEdge({\bf START}, $j$, $0$)
\FOR{$i$ := $1$ to $j-1$}
\IF{$l_j\geq0$ AND $l_i\geq0$}
\IF{Source-Target$\left(\overline{\text{CNOT}}(a_{i},b_{i})(D^{l_i}),\overline{ \text{CNOT}}(a_{j},b_{j})(D^{l_j})\right) = {\bf TRUE}$}
\STATE DrawEdge($i,j,l_i$)
\ELSIF{Target-Source$\left(\overline{\text{CNOT}}(a_{i},b_{i})(D^{l_i}),\overline{ \text{CNOT}}(a_{j},b_{j})(D^{l_j})\right) = {\bf TRUE} $}
\STATE DrawEdge($i,j,-l_j$)
\ENDIF
\ELSIF{$l_j\geq0$ AND $l_i < 0$}
\IF{Source-Target$\left(\overline{\text{CNOT}}(a_{i},b_{i})(D^{l_i}),\overline{ \text{CNOT}}(a_{j},b_{j})(D^{l_j})\right) = {\bf TRUE} $}
\STATE DrawEdge($i,j,0$)
\ENDIF
\IF{Target-Source$\left(\overline{\text{CNOT}}(a_{i},b_{i})(D^{l_i}),\overline{ \text{CNOT}}(a_{j},b_{j})(D^{l_j})\right) = {\bf TRUE} $}
\STATE DrawEdge($i,j,|l_i|-l_j$)
\ENDIF
\ELSIF{$l_j < 0$ AND $l_i \geq 0$}
\IF{Source-Target$\left(\overline{\text{CNOT}}(a_{i},b_{i})(D^{l_i}),\overline{ \text{CNOT}}(a_{j},b_{j})(D^{l_j})\right) = {\bf TRUE} $}
\STATE DrawEdge($i,j,l_i-|l_j|$)
\ENDIF
\IF{Target-Source$\left(\overline{\text{CNOT}}(a_{i},b_{i})(D^{l_i}),\overline{ \text{CNOT}}(a_{j},b_{j})(D^{l_j})\right) = {\bf TRUE} $}
\STATE DrawEdge($i,j,0$)
\ENDIF
\ELSIF{$l_j < 0$ AND $l_i < 0$}
\IF{Target-Source$\left(\overline{\text{CNOT}}(a_{i},b_{i})(D^{l_i}),\overline{ \text{CNOT}}(a_{j},b_{j})(D^{l_j})\right) = {\bf TRUE} $}
\STATE DrawEdge($i,j,|l_i|$)
\ELSIF{Source-Target$\left(\overline{\text{CNOT}}(a_{i},b_{i})(D^{l_i}),\overline{ \text{CNOT}}(a_{j},b_{j})(D^{l_j})\right) = {\bf TRUE} $}
\STATE DrawEdge($i,j,-|l_j|$)
\ENDIF
\ENDIF
\ENDFOR
\ENDFOR
\STATE Draw an {\bf END} vertex.
\FOR{$j$ := $1$ to $N$}
\STATE DrawEdge($j$, {\bf END}, $|l_j|$)
\ENDFOR
\end{algorithmic}
\end{algorithm}

\subsubsection{The longest path gives the minimal memory requirements}

Theorem~\ref{thm:final-theorem} states that the
weight of the longest path from the START vertex to END vertex in $G$ is equal
to the minimal required memory for the encoding sequence in (\ref{eq:mixed}).
Part of its proof follows from Lemma~\ref{lem:final-lemma}.

\begin{theorem}
\label{thm:final-theorem}The weight of the longest path from the START vertex
to the END vertex in the commutativity graph $G$ is equal to the memory
required for the convolutional encoder.
\end{theorem}

\begin{proof}
The longest path in the graph is the maximum of the longest path to each
vertex summed with the weight of the edge from each vertex to the end vertex:%
\[
w=\max\{w_{j}+|l_{j}|\}_{j\in\{1,2,\cdots,N\}},
\]
The following relation holds by applying Lemma~\ref{lem:final-lemma} below:%
\[
w=\max\{\{\mathbb{\tau}_{j}+l_{j}\}_{l_{j}\geq0},\{\mathbb{\sigma}_{j}%
+|l_{j}|\}_{l_{j}<0}\},
\]
which is equal to%
\[
w=\max\{\{\mathbb{\sigma}_{j}\}_{l_{j}\geq0},\{\mathbb{\tau}_{j}\}_{l_{j}%
<0}\}.
\]
The quantity $\max\{\{\mathbb{\sigma}_{j}\}_{l_{j}\geq0},\{\mathbb{\tau}%
_{j}\}_{l_{j}<0}\}$ is equal to the memory requirement of a minimal-memory
convolutional encoder, so the theorem holds.
\end{proof}

\begin{lemma}
\label{lem:final-lemma}The weight of the longest path from the START vertex to
vertex $j$ in $G,$ $w_{j}$ is equal to
\[
w_{j}=\left\{
\begin{array}
[c]{ccc}%
\tau_{j} & : & l_{j}\geq0\\
\sigma_{j} & : & l_{j}<0
\end{array}
\right.  .
\]

\end{lemma}

\begin{proof}
We prove the lemma by induction. The weight $w_{1}$ of the path to the first
vertex is equal to zero because a zero-weight edge connects the START vertex
to the first vertex. If $l_{1}\geq0$, then $\tau_{1}=0$. So $w_{1}=\tau_{1}=0$
and if $l_{1}<0$, then $\sigma_{1}=0$. So $w_{1}=\sigma_{1}=0$. Therefore the
lemma holds for the first gate.

Suppose the lemma holds for the first $k$ gates:%
\[
\forall i\in\{1,\cdots,k\},\ \ \ \ w_{i}=\left\{
\begin{array}
[c]{ccc}%
\tau_{i} & : & l_{i}\geq0\\
\sigma_{i} & : & l_{i}<0
\end{array}
\right.  .
\]
Consider adding a $(k+1)^{\text{th}}$ gate string $\overline{\text{CNOT}}(a_{k+1}%
,b_{k+1})(D^{l_{k+1}})$ to the pearl-necklace encoder.
Algorithm~\ref{alg:nc-graph-mixed} then adds a vertex with label $k+1$ to the
graph. First consider the case that $l_{k+1}\geq0$.
Algorithm~\ref{alg:nc-graph-mixed} then adds the following edges to the graph
$G$:

\begin{enumerate}
\item A zero-weight edge from the START vertex to vertex $k+1$.

\item An $l_{k+1}$-weight edge from vertex $k+1$ to the END vertex.

\item An $l_{i}$-weight edge from each vertex $\{i\}_{i\in\mathcal{S}%
_{k+1}^{+}}$ to vertex $k+1$.

\item A zero-weight edge from each vertex $\{i\}_{i\in\mathcal{S}_{k+1}^{-}}$
to vertex $k+1$.

\item A $-l_{k+1}$-weight edge from each vertex $\{i\}_{i\in\mathcal{T}%
_{k+1}^{+}}$ to vertex $k+1$.

\item A $\left\vert l_{i}\right\vert -l_{k+1}$-weight edge from each vertex
$\{i\}_{i\in\mathcal{T}_{k+1}^{-}}$ to vertex $k+1$.
\end{enumerate}

The weight of the longest path from the START vertex to vertex $k+1$ is then
as follows:
\begin{align}
w_{k+1} &  =\max\{0,\{w_{i}+l_{i}\}_{i\in\mathcal{S}_{k+1}^{+}},\{w_{i}%
\}_{i\in\mathcal{S}_{k+1}^{-}},\{w_{i}-l_{k+1}\}_{i\in\mathcal{T}_{k+1}^{+}%
},\{w_{i}+|l_{i}|-l_{k+1}\}_{i\in\mathcal{T}_{k+1}^{-}}\}\nonumber\\
&  =\max\{0,\{\mathbb{\tau}_{i}+l_{i}\}_{i\in\mathcal{S}_{k+1}^{+}%
},\{\mathbb{\sigma}_{i}\}_{i\in\mathcal{S}_{k+1}^{-}},\{\mathbb{\tau
}_{i}-l_{k+1}\}_{i\in\mathcal{T}_{k+1}^{+}},\{\mathbb{\sigma}_{i}%
+|l_{i}|-l_{k+1}\}_{i\in\mathcal{T}_{k+1}^{-}}%
\}.\label{eq:G-pos-tar-constraint}%
\end{align}
The following relation follows by applying (\ref{eq:tar-mixed-cons}) and
(\ref{eq:G-pos-tar-constraint}) when $l_{k+1}\geq0$:%
\[
w_{k+1}=\tau_{k+1},
\]
In a similar way, we can prove that%
\[
w_{k+1}=\sigma_{k+1},
\]
when $l_{k+1}<0$ and this last step concludes the proof.
\end{proof}

The complexity of constructing the graph $G$ is $O(N^{2})$ (the argument is
similar to before), and dynamic programming finds the longest path in $G$ in
time linear in the number of its vertices and edges because $G$\ is a
weighted, directed acyclic graph.
\subsubsection{Example of a pearl-necklace encoder with arbitrary CNOT gates}

We conclude the final development with an example.
\begin{figure}[ptb]
\begin{center}
\includegraphics[ natheight=6.679800in, natwidth=11.313500in,
width=5.086in ]
{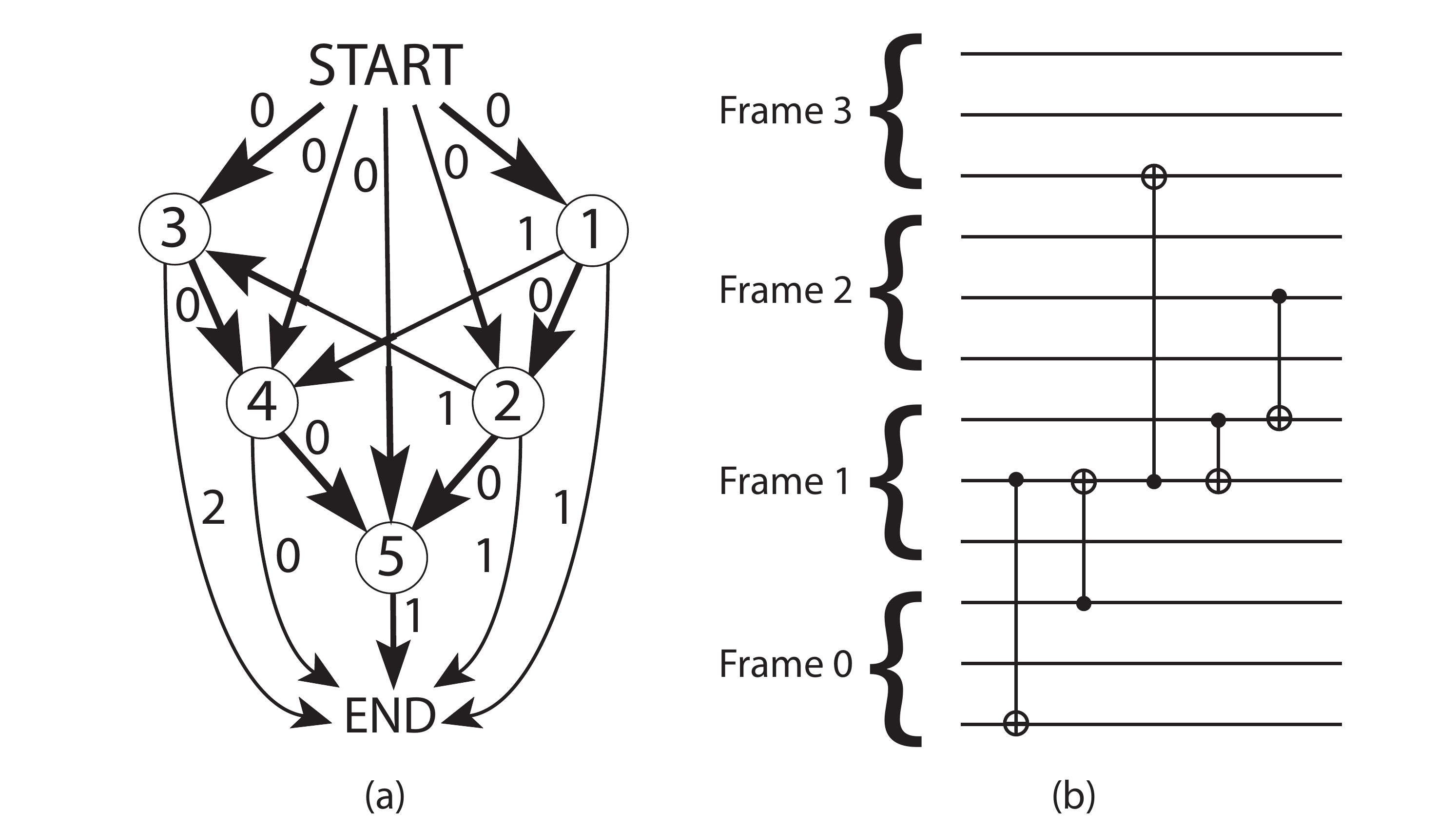}
\end{center}
\caption{(a) The commutativity graph $G^{+}$ and (b) a minimal-memory
convolutional encoder for Example~\ref{ex:mixed}.}%
\label{fig:example-3}%
\end{figure}
\begin{example}
\label{ex:mixed}Consider the following succession of gate strings in a pearl-necklace encoder:%
\[
\overline{\text{CNOT}}(2,3)(D)\ \overline{\text{CNOT}}(1,2)(D^{-1}%
)\ \overline{\text{CNOT}}(2,3)(D^{-2})\ \overline{\text{CNOT}}%
(1,2)(1)\ \overline{\text{CNOT}}(2,1)(D).
\]
Figure \ref{fig:example-3}(a) illustrates $G$ for the above example. The
longest path is
\[
\text{START}\rightarrow2\rightarrow3\rightarrow\text{END}%
\]
with its weight equal to three. Thus, the minimal-memory convolutional encoder
requires three frames of memory qubits. Also, from inspecting the graph $G$, we can determine the source qubit and target qubit frame indices in the convolutional encoder:
$\tau_{1}=0$, $\sigma_{2}=0$, $\sigma_{3}=1$, $\tau=1$, and $\tau_{5}=1$.
Figure~\ref{fig:example-3}(b) depicts a minimal-memory convolutional encoder.
\end{example}

\section{Conclusion}

We have shown how to realize a minimal-memory convolutional encoder that
implements the same transformation as a pearl-necklace encoder with arbitrary
CNOT\ gate strings. Our approach is to construct a dependency graph whose directed
edges represent non-commutative relations between gate strings in
the pearl-necklace encoder. Determining the minimal memory is then the same
task as determining the longest path through this graph. The algorithm for
constructing and searching the graph requires time at worst quadratic
in the number of gate strings in the pearl-necklace encoder. This technique
should be useful when we have a pearl-necklace encoder description, which is
the case in the work of Grassl and
R\"{o}tteler~\cite{isit2006grassl,ieee2007grassl} and later work on
entanglement-assisted quantum convolutional coding~\cite{arx2007wildeEAQCC}.\cite{arx2007wildeCED,arx2007wildeEAQCC,arx2008wildeUQCC,arx2008wildeGEAQCC,pra2009wilde}

A later paper includes the general case of the algorithm for
pearl-necklace encoders with gate strings other than CNOT\ gate strings~\cite{general-paper}. The
extension of the algorithm will include all gate strings that are in the
shift-invariant Clifford group~\cite{isit2006grassl,PhysRevA.79.062325},
including Hadamard gates, phase gates, two variants of the controlled-phase
gate string, and infinite-depth CNOT operations.

There might be ways to determine convolutional encoders with even lower memory
requirements by using techniques different from those given here. First, our algorithm begins with
a particular pearl-necklace encoder, i.e., a particular succession of gate
strings to implement. One could first perform an optimization over all
possible pearl-necklace encoders of a quantum convolutional code because there
are many pearl-necklace encoders for a particular quantum convolutional code.
Perhaps even better, one could look for a method to construct a repeated
unitary directly from the polynomial description of the code itself. Ollivier
and Tillich have considered such an approach in Sections~2.3 and 3 of
Ref.~\cite{arxiv2004olliv}, but it is not clear that their technique is
attempting to minimize the memory resources for the encoder. There are
well-developed techniques in the classical world to determine minimal-memory
encoders. Ideally, we would like to have a \textquotedblleft
quantization\textquotedblright\ of Theorem~2.22 of Ref.~\cite{book1999conv}.

% \section*{Acknowledgements}

The authors acknowledge useful discussions with Patrick Hayden and Pranab Sen.
MMW acknowledges support from the MDEIE (Qu\'{e}bec) PSR-SIIRI international collaboration grant.

\begin{IEEEbiographynophoto}{Monireh Houshmand} was born in Mashhad, Iran. She got her B.S. and M.S. degrees in Electrical Engineering from Ferdowsi University of Mashhad, in 2005 and 2007, respectively. She is currently a Ph.D. student of electrical engineering in Ferdowsi University of Mashhad. Her research interests include quantum error correction and quantum cryptography.
\end{IEEEbiographynophoto}

\begin{IEEEbiographynophoto}{Saied Hosseini Khayat} (Member, IEEE) was born in Tehran, Iran. He received the BS degree in electrical engineering from Shiraz University, Iran (1986), and the MS and PhD degrees in electrical engineering from Washington University in St. Louis, USA (1991 and 1997). He held R\&D positions at Globespan Semiconductor
and Erlang Technologies.
His expertise and interests lie in the area of digital system design for communications, networking and cryptography. At present, he is an assistant professor of electrical engineering at Ferdowsi University of Mashhad.%, Iran.
\end{IEEEbiographynophoto}

\begin{IEEEbiographynophoto}{Mark M. Wilde} (M'99) was born in Metairie, Louisiana.
He obtained a B.S.~in computer engineering from
Texas A\&M University in 2002,
an M.S.~in electrical engineering from Tulane University in 2004, and the Ph.D.~in electrical engineering
from the University of Southern California
in 2008. Presently, he is a postdoctoral fellow in the School of Computer Science at
McGill University in Montreal.
His current research interests are in
quantum Shannon theory and quantum error correction.
\end{IEEEbiographynophoto}

\end{document}